\documentclass[sigconf,screen]{acmart}
\pdfoutput=1

\setcopyright{none}
\acmYear{2023}
\acmDOI{}
\acmPrice{}
\acmISBN{}

\acmConference[QP2023]{The First International Workshop on the Art, Science, 
and Engineering of Quantum Programming}{March 13--14}{Tokyo, Japan}
\usepackage{graphicx,wrapfig,color,colortbl,bm,sidecap,xspace}
\usepackage[export]{adjustbox}
\usepackage{amsmath}
 
\usepackage{amssymb}
\usepackage{braket}
\usepackage[labelfont=bf]{caption}
\usepackage{color}
\usepackage{colortbl}
\usepackage{float}
\usepackage{graphicx}
\usepackage{mathtools}
\usepackage{pgfplots}
\usepackage{physics}
\usepackage{stmaryrd}
\SetSymbolFont{stmry}{bold}{U}{stmry}{m}{n}
\usepackage{subcaption}
\usepackage[normalem]{ulem}
\usepackage{scalerel}
\usepackage{textcomp}
\usepackage{hyperref}

\usepackage{tikzit}

\tikzstyle{box}=[shape=rectangle, text height=1.5ex, text depth=0.25ex, yshift=0.5mm, fill=white, draw=black, minimum height=5mm, yshift=-0.5mm, minimum width=5mm, font={\small}]
\tikzstyle{Z dot}=[inner sep=0mm, minimum size=2mm, shape=circle, draw=black, fill={rgb,255: red,216; green,248; blue,216}, tikzit fill={rgb,255: red,216; green,248; blue,216}]
\tikzstyle{Z phase dot}=[minimum size=4.75mm, font={\footnotesize}, shape=rectangle, rounded corners=1.9mm, inner sep=0.1mm, outer sep=-2mm, scale=0.8, tikzit shape=circle, draw=black, fill={rgb,255: red,216; green,248; blue,216}, tikzit draw=blue, tikzit fill={rgb,255: red,216; green,248; blue,216}]
\tikzstyle{X dot}=[Z dot, shape=circle, draw=black, fill={rgb,255: red,232; green,165; blue,165}, tikzit fill={rgb,255: red,232; green,165; blue,165}]
\tikzstyle{X phase dot}=[Z phase dot, tikzit shape=circle, tikzit fill={rgb,255: red,232; green,165; blue,165}, fill={rgb,255: red,232; green,165; blue,165}, font={\footnotesize}, tikzit draw=blue]
\tikzstyle{XD dot}=[shape=XDdot, inner sep=2pt, draw=black, tikzit fill={rgb,255: red,255; green,191; blue,191}]
\tikzstyle{XD phase dot}=[shape=XDdotphase, minimum size=4.75mm, font={\footnotesize}, inner sep=.1mm, outer sep=0mm, scale=0.8, tikzit shape=circle, rounded corners=1.9mm, draw=black, tikzit fill={rgb,255: red,255; green,191; blue,191}, tikzit draw=blue]
\tikzstyle{zn}=[shape=zn, tikzit draw=black, draw=black, inner sep=2pt]
\tikzstyle{hadamard}=[fill=white, draw=black, shape=rectangle, inner sep=0mm, minimum height=2.75mm, minimum width=2.75mm]
\tikzstyle{hz}=[hadamard, fill={rgb,255: red,216; green,248; blue,216}, shape=rectangle, tikzit fill={rgb,255: red,216; green,248; blue,216}, minimum height=2 mm, minimum width=1.25 mm, tikzit draw=black]
\tikzstyle{hx}=[hadamard, fill={rgb,255: red,232; green,165; blue,165}, shape=rectangle, tikzit fill={rgb,255: red,232; green,165; blue,165}, minimum height=2 mm, minimum width=1.25 mm, tikzit draw=black]
\tikzstyle{vertex}=[inner sep=0mm, minimum size=1mm, shape=circle, draw=black, fill=black]
\tikzstyle{vertex set}=[inner sep=0mm, minimum size=1mm, shape=circle, draw=black, fill=white, font={\footnotesize\boldmath}]
\tikzstyle{meter}=[draw, fill=white, minimum width=2em, minimum height=1.5em, rectangle, path picture={\draw ([shift={(.1,.24)}]path picture bounding box.south west) to[bend left=50] ([shift={(-.1,.24)}]path picture bounding box.south east);\draw[-{Latex[scale=0.6]}] ([shift={(0,.1)}]path picture bounding box.south) -- ([shift={(.3,-.1)}]path picture bounding box.north);}, tikzit shape=rectangle]
\tikzstyle{white dot}=[Z dot]
\tikzstyle{gray dot}=[X dot]
\tikzstyle{white phase dot}=[Z phase dot]
\tikzstyle{gray phase dot}=[X phase dot]
\tikzstyle{red ket}=[fill={rgb,255: red,232; green,165; blue,165}, draw=black, shape=isosceles triangle, tikzit fill={rgb,255: red,232; green,165; blue,165}, tikzit draw=black, inner sep=0 mm, outer sep=2 mm]
\tikzstyle{tiny none}=[none, font={\tiny}]
\tikzstyle{green ket}=[fill={rgb,255: red,216; green,248; blue,216}, draw=black, shape=isosceles triangle, tikzit fill={rgb,255: red,216; green,248; blue,216}, tikzit draw=black, inner sep=0 mm, outer sep=2 mm]
\tikzstyle{filament}=[hadamard, fill=yellow, draw=none, minimum height=0.01mm]
\tikzstyle{unit circle}=[shape=circle, minimum size=42.5 mm, fill=none, draw={rgb,255: red,223; green,223; blue,223}, tikzit draw={rgb,255: red,223; green,223; blue,223}]
\tikzstyle{new style 0}=[shape=ellipse, minimum height=10 mm, minimum width=100 mm, fill=white, draw=black]
\tikzstyle{gate}=[box, minimum height=10mm, minimum width=10mm]
\tikzstyle{2 control}=[vertex set, draw=blue, inner sep=0.5pt]
\tikzstyle{1 control}=[2 control, draw=red]
\tikzstyle{0 control}=[2 control, draw=black]
\tikzstyle{small dot}=[vertex, minimum size=1 mm, draw=black, tikzit draw=black, tikzit fill=black, tikzit shape=circle]
\tikzstyle{tallbox}=[box, minimum height=12mm]
\tikzstyle{targ}=[vertex set, minimum size=0.5mm, inner sep=-0.5mm, tikzit shape=circle, shape=circle, tikzit draw=black]
\tikzstyle{hadamardbox}=[hadamard, xslant=0]
\tikzstyle{pink dot}=[X dot, fill={rgb,255: red,255; green,200; blue,240}]
\tikzstyle{pink phase dot}=[Z phase dot, fill={rgb,255: red,255; green,200; blue,240}]

\tikzstyle{directedarrow}=[draw={rgb,255: red,223; green,223; blue,223}, ->, tikzit draw={rgb,255: red,223; green,223; blue,223}, line width=1 pt]
\tikzstyle{simple}=[-]
\tikzstyle{hadamard edge}=[-, color={rgb,255: red,0; green,100; blue,248}, dashed, dash pattern=on 2pt off 0.7pt, tikzit draw={rgb,255: red,0; green,100; blue,248}]
\tikzstyle{brace edge}=[-, tikzit draw=blue, decorate, decoration={brace,amplitude=1mm,raise=-1mm}]
\tikzstyle{gray}=[-, draw={rgb,255: red,223; green,223; blue,223}, line width=1 pt]
\tikzstyle{arrow}=[<-, draw={rgb,255: red,128; green,128; blue,128}]
\tikzstyle{double-arrow}=[draw={rgb,255: red,128; green,128; blue,128}, <->]
\tikzstyle{dashed edge}=[-, dashed, dash pattern=on 2pt off 0.5pt, draw=black]
\tikzstyle{diredge}=[->]
\tikzstyle{double edge}=[-, double, shorten <=-1mm, shorten >=-1mm, double distance=2pt]
\tikzstyle{thin}=[-, line width=0.05mm]
\tikzstyle{thin gray}=[-, draw={rgb,255: red,223; green,223; blue,223}, line width=0.05mm]
\tikzstyle{less thin}=[-, line width=0.1mm]
\tikzstyle{dashed gray edge}=[-, dashed edge, draw={rgb,255: red,128; green,128; blue,128}]
\tikzstyle{light right directed arrow}=[->, directedarrow, draw={rgb,255: red,223; green,223; blue,223}, line width=0.2mm]
\tikzstyle{diredge0.3}=[->, line width=0.3 mm]
\tikzstyle{less thin gray}=[-, draw={rgb,255: red,223; green,223; blue,223}]
\tikzstyle{dashed thin purple}=[-, dashed, line width=0.1mm, draw={rgb,255: red,128; green,106; blue,219}]
\tikzstyle{hadamardedge}=[-, color={rgb,255: red,100; green,200; blue,248}, dashed, dash pattern=on 2pt off 0.7pt, tikzit draw={rgb,255: red,120; green,220; blue,248}]
\tikzstyle{line0.3}=[-, line width=0.3mm]
\tikzstyle{light blue line}=[-, color={rgb,255: red,100; green,200; blue,248}, tikzit draw={rgb,255: red,100; green,200; blue,248}]
\tikzstyle{pink0.2}=[-, color={rgb,255: red,248; green,100; blue,200}, line width=0.2mm, tikzit draw={rgb,255: red,248; green,100; blue,200}]

\input{zx.tikzdefs}

\DeclareUnicodeCharacter{2212}{-}
\newcommand{\ceil}[1]{\left\lceil #1 \right\rceil}
\renewcommand\vec[1]{\overrightarrow{#1}}

\newcommand{\minu}{\raisebox{-1pt}{\texttt{-}}}

\begin{document}

\title{Scaling W state circuits in the qudit Clifford hierarchy}

\author{Lia Yeh}
\orcid{0000-0003-2704-4057}
\affiliation{%
    \institution{Department of Computer Science, University of Oxford}
    \city{Oxford}
    \country{United Kingdom}
}
\affiliation{%
    \institution{Quantinuum Ltd.}
    \city{17 Beaumont Street, Oxford OX1 2NA}
    \country{United Kingdom}
}
\email{lia.yeh@cs.ox.ac.uk}

\begin{abstract}
    We identify a novel qudit gate which we call the $\sqrt[d]{Z}$ gate.
    This is an alternate generalization of the qutrit $T$ gate to any odd prime dimension $d$, in the $d^{\text{th}}$ level of the Clifford hierarchy.
    Using this gate which is efficiently realizable fault-tolerantly should a certain conjecture hold, we deterministically construct in the Clifford+$\sqrt[d]{Z}$ gate set, $d$-qubit $W$ states in the qudit $\{\ket{0},\ket{1}\}$ subspace.  For qutrits, this gives deterministic and fault-tolerant constructions for the qubit $W$ state of sizes three with $T$ count 3, six, and powers of three.

    Furthermore, we adapt these constructions to recursively scale the $W$ state size to arbitrary size $N$, in $O(N)$ gate count and $O(\text{log }N)$ depth.  This is moreover deterministic for any size qubit $W$ state, and for any prime $d$-dimensional qudit $W$ state, size a power of $d$.
    
    For these purposes, we devise constructions of the $\ket{0}$-controlled Pauli $X$ gate and the controlled Hadamard gate in any prime qudit dimension.  These decompositions, for which exact synthesis is unknown in Clifford+$T$ for $d > 3$, may be of independent interest.
\end{abstract}



\maketitle

\section{Introduction}
Understanding the structure of multipartite entanglement has long been of fundamental interest.  The two multipartite entangled states most important in the literature are the $GHZ$ state and the $W$ state:
\begin{definition}\label{def:qubitWstate}
    The $N$-qubit $W$ state is
    \begin{equation}
        \ket{W_N} = \frac{1}{\sqrt{N}} \left( \ket{10...0} + \ket{010...0} + ... + \ket{0...01} \right).
    \end{equation}
\end{definition}
Characterization of tripartite entangled states by D\"ur, Vidal, and Cirac found there to be only two distinct SLOCC-equivalence classes (equivalent up to stochastic local operations and classical communication): the $GHZ$-class and the $W$-class~\cite{DurW2000tripartiteclasses}.  As three-qubit states, $GHZ$ states and the $W$ states correspond respectively to a special and an anti-special Frobenius algebra, defining a compositional structure which can be used to compose arbitrarily large multipartite qubit states from smaller ones~\cite{CoeckeB2010multipartite}.

In contrast to the $GHZ$ state, the entanglement of the generalized $W$ state is highly robust against decoherence and particle loss~\cite{ZhaoMJ2010mixedWstatesmonogamy}.
Robust and scalable multipartite entanglement is essential to many applications of quantum technology, including quantum communication, encryption, and distributed quantum computing.  The W state in particular has been employed in protocols such as anonymous transmission in a quantum network~\cite{LipinskaV2018anonnetW}, photonic error detection~\cite{VijayanMK2020Werrordetect}, quantum memory~\cite{LiC2020quantummem}, quantum secret sharing~\cite{TsaiCW2019Wsecretsharing}, and deterministic quantum communication~\cite{TsaiCW2013detqcomm}.
A notable example is the classical distributed computing task of leader election, which is commonly prerequisite for graph-theoretic tasks such as computing minimum spanning trees.  No known classical algorithm satisfies the same conditions achievable by quantum leader election algorithms for anonymous networks~\cite{vanMeter2014quantumnetworking}, for which (unitary transforms of) W states are not only sufficient, but necessary for~\cite{D'Hondt2006wandghz}.

A number of proposals exist in the literature for building larger (more than three qubits) W states.
Some architectures are more suited for generating large W states directly, with proposed experimental schemes including single photon 8-qubit W state~\cite{HeilmannR2015w8} and arbitrary size perfect W states (asymmetric as to be better suited to teleportation and superdense coding)~\cite{SwainM2022nWperfect}.
One approach~\cite{GrafeM2014photonWstates,DikerF2016detconstructWstates} is analogous to reservoir sampling, where a series of different sized rotations distributes the entanglement uniformly with quadratic gate count, which could be made asymptotically linear with a rather involved binary-to-unary conversion~\cite{GidneyC2018Wstatereservoirsampling}.
In another vein, two approaches to generate larger W states from smaller W states, are by fusing two smaller W states together~\cite{OzdemirS2011fuseW, BuguS2013WfuseFredkin, LiK2016fuseW} for instance in a cavity QED architecture~\cite{ZangXP2015fuseWQED,JiYQ2017fuseWQED}, or by expanding a W state with ancillae~\cite{LiK2016fuseW}.
In addition, a technique compatible with the above methods is entanglement concentration protocols~\cite{BennettC1996concentrateentanglement} which concentrate a less entangled W state on an arbitrary number of qubits into a maximally entangled W state~\cite{ShengY2015nWparitycheck}.

Although these techniques can generate arbitrarily large (maximally entangled) W states in principle, their scalability has a practical limitation due to physical errors.
These protocols were not intended for compatibility with quantum error correction codes and therefore are prone to propagation of errors, an issue that leads to their performance falling off with increasingly large system sizes.  This can be especially concerning if the use case pertains to fairness (for instance quantum algorithms for voting~\cite{BaoN2017voting}) or cryptography (for instance quantum key distribution~\cite{WangHW2022qkdW}).

As a result, there is a need to investigate improved fault-tolerant protocols to prepare these highly entangled resource states.
In this setting, the most probable errors on a subset of operations can be efficiently correctable.
This transition necessitates a paradigm shift in the requirements and cost model that must be considered, to minimize the errors incurred by non-correctable operations.
To the best of our knowledge, there has been limited investigation into fault-tolerant W state generation beyond a StackExchange post in 2018 by Gidney~\cite{GidneyC2018Wstatemagic}.

In this work, we augment fault-tolerant protocols to construct large W states deterministically (i.e. with $100\%$ success probability of the protocol in theory, independent of physical errors) in linear gate count and logarithmic depth.
We do this by leveraging the higher levels of qudits as a computational resource, thus removing the need for ancillae and post-selection for constructing certain size W states.
The higher levels of many qubit devices are addressable and support operations making them viable as \emph{qudits} --- the generalization of qubits (for which $d=2$) to $d$-dimensional Hilbert spaces.
These enable applications including qudit algorithms~\cite{GedikZ2015quditalgperm,WangY2020quditsreview}, improved magic state distillation noise thresholds~\cite{CampbellE2014quditmsdthresholds}, communication noise resilience~\cite{CozzolinoD2019quditcommunication}, Bell inequality violation~\cite{VertesiT2010bellxpmtqudits}, testing quantum gravity~\cite{TillyJ2021qgravity}, and emulation of binary circuits~\cite{LanyonBP2009quditemu, GokhaleP2019asymptotic,NikolaevaAS2022mctqutrit}.  Moreover, qubit circuit emulation of an arbitrary qudit unitary has been argued to be less asymptotically efficient than qudit circuit decomposition (except in the case where $d$ is a power of 2)~\cite{BullockS2005asympoptqudit}.

Qudit computation has been proposed or experimentally demonstrated on a number of quantum processors~\cite{MatoK2022adaptivecompilequdits} including cold atoms~\cite{AndersonBE2015quditcoldatoms,KasperV2022quditcoldatoms}, nuclear magnetic resonance~\cite{GedikZ2015quditalgperm}, nuclear spins~\cite{GodfrinC2017quditnuclearspins}, photonic quantum computing~\cite{RingbauerM2018ququartphotonic,HuXM2018ququart}, Rydberg atoms~\cite{WeggemansJ2022quditRydberg}, superconducting quantum computing~\cite{CaoS2023ququartemuvqe,BlokM2021scrambling,YeB2018cphasephoton,YurtalanM2020Walsh-Hadamard,HillA2021doublycontrolled}, and trapped ions~\cite{RingbauerM2022quditions}.  As of recent, error rates of universal qudit processors were shown to be competitive relative to qubit processors~\cite{RingbauerM2018ququartphotonic,ChiY2022quditphotonicprocessor}, and maximally entangling qudit gates have been physically realized~\cite{HrmoP2023quditions}.

Compared to the qubit case, the entanglement properties of qudits are less understood.  Experimentally, qudit entanglement has demonstrated improved noise robustness compared to qubits~\cite{SrivastavV2022noiserobustquditsteering}, and leveraging qutrit states can improve fidelity of qubit operations~\cite{BaekkegaardT2019scqubitqutrit,BrownT2022qutritqubitentangle}.
With regards to theory, although the $W$ class has been generalized to qudits with similar properties of monogamy relations~\cite{KimJS2008Wclass}, entanglement classification of multiple qudit states is a hard task made more difficult by the existence of a form of entanglement that does not exist for qubits~\cite{PoppC2023boundEntanglementQuquarts}.
\begin{definition}\label{def:quditWstate}
    The $N$-qudit $W$ state is often defined as~\cite{KimJS2008Wclass}
    \begin{equation}
        \ket{W_N} = \frac{1}{\sqrt{(d\minu 1) N}} \sum_{j=1}^{d\minu 1} \left( \ket{j0...0} + \ket{0j0...0} + ... + \ket{0...0j} \right).
    \end{equation}
\end{definition}

The gate set this work is set in is what we call the $\sqrt[d]{Z}$ (can also be denoted $Z^{(1/d)}$ or $dZ$) gate.  This is a family of gates, one for each prime qudit (i.e. qupit) dimension $d$, in the $d^{\text{th}}$ level of the Clifford hierarchy.  For $d = 3$ this coincides with the qutrit $T$ gate.  Indeed, as we show in Section~\ref{sec:cliffordplust}, this can be seen as an alternate generalization of the qutrit $T$ gate to all odd prime dimensions, to the characterization of single-qupit diagonal gates in the third level of the Clifford hierarchy by Howard and Vala~\cite{HowardM2012quditTgate}.

Despite the yet to be determined practicality of fault-tolerant implementation of the $\sqrt[d]{Z}$ gate, we claim that it is important to study it given that it enables certain logical constructions of computational importance.  A number of essential constructions in the qubit setting have no known implementation applicable in arbitrary prime dimension.  In this work we focus on constructing $W$ states considering the asymptotic gate count of a specific non-Clifford gate.  In forthcoming work we will demonstrate that this gate set enables exact synthesis of any classical reversible logic gate (including any qupit multiple-controlled Toffoli).  In prior work with van de Wetering, we presented circuit constructions in qutrit Clifford+$T$ for not only any ternary classical reversible gate, but also any multiple-controlled qutrit Clifford+$T$ gate~\cite{YehL2022controlunivqutrit}.

This can also facilitate future decompositions in other gate sets by establishing a pathway to realizing these computationally useful gates, through showing that direct implementation of the $\sqrt[d]{Z}$ gate is a sufficient condition.
We note that proving whether a gate is impossible to exactly synthesize in a certain gate set is difficult, and is oftentimes preceded by the generally difficult problem of full characterization of all gates exactly synthesizable in that gate set.

The paper is organized as follows.  In Section~\ref{sec:background}, we introduce the qupit gates present in this work and whether they are Pauli, Clifford, or higher in the Clifford hierarchy.
The main body of the paper is in Section~\ref{sec:controlledandW}.  This contains circuit decompositions of the $\ket{0}$-controlled $X$ gate and the $\ket{0}$-controlled $H$ gate in Sections~\ref{sec:ZCX}~and~\ref{sec:CH} respectively.  These are component gates in the circuit constructions of $d$-qubit $W$ states in Section~\ref{sec:Wdstate}, and $d^n$-qubit and $d^n$-qudit $W$ states in Section~\ref{sec:primepowerW} for $n \in \mathbb{Z}^+$.  We then discuss extending these protocols to any size $W$ states in Section~\ref{sec:othersizeW} by mixed dimensional computation and by post-selection, before brief comment on network topology in Section~\ref
{sec:network}.  We end with concluding remarks and future work in Section~\ref{sec:conclusion}.

\section{Background}\label{sec:background}
\subsection{Qupit Cliffords}
Several concepts for qubits extend to qutrits, or more generally to qu\emph{dits}, which are $d$-dimensional quantum systems --- in particular, the concept of Pauli's and Cliffords.
For a $d$-dimensional qudit, we follow the standard qudit definition of the Pauli $Z$ and $X$ gates, colloquially called the clock and shift operators, as~\cite{GottesmanD1999ftqudit,HowardM2012quditTgate}
\begin{equation}
    Z\ket{k} = \omega^k \ket{k} \qquad\qquad X\ket{k} = \ket{k+1}
\end{equation}
where $k \in \mathbb{Z}_d$.  Here $\omega \coloneqq e^{2\pi i/d}$ is such that $\omega^d = 1$, and the addition $\ket{k+1}$ is taken modulo $d$.  For this reason, in circuit notation we adopt the shorthand of $+1$ being the $X$ gate and $+k$ being the $X^k$ gate.  The \emph{Pauli group} is the set of unitaries generated by tensor products of the $X$ and $Z$ gate.

Another concept that readily generalizes to qudits is that of Clifford unitaries.
\begin{definition}\label{def:Clifford}
    Let $U$ be a qudit unitary acting on $n$ qudits. We say it is \emph{Clifford} when every Pauli is mapped to another Pauli under conjugation by $U$, i.e.~if $UPU^\dagger$ is in the Pauli group for any Pauli $P$.
\end{definition}
The set of $n$-qudit Cliffords forms a group under composition. For qubits, this group is generated by the $S$, $H$ and $CX$ gates. The same is true for qupits, for the right generalization of these gates\footnote{The gate definitions for various qudit Cliffords may vary across the literature up to a global phase.  Indeed, by Definition~\ref{def:Clifford}, whether a gate is Clifford is invariant under changes in global phase.}.

For qubits the $X$ gate is the NOT gate, while $Z=\text{diag}(1,-1)$. For a qubit there is only one non-trivial permutation of the standard basis states, implemented by the $X$ gate.  For qutrits, there are five ternary $X$ gates, as there are five non-trivial permutations of the $Z$-basis states: 2 of which are Pauli $X$ and $X^\dagger$, and the other 3 exchange exactly 2 of the 3 qutrit $Z$-basis states.  As the qupit dimension $d$ increases, a growing number of permutations of single-qupit $Z$-basis states fall outside the single-qupit Clifford unitaries.

Just as the $Z$-basis states labelled by elements of $\mathbb{Z}_d$ are eigenstates of the $Z$ operator, the $X$-basis states are eigenstates of the $X$ operator.  We denote the $\ket{+}$ state as the $X$-basis state with eigenvalue $1$,
\begin{equation}
    \ket{+} = \frac{1}{\sqrt{d}} \sum_{j=0}^{d\minu 1} \ket{j}.
\end{equation}

We now define the Clifford gate set, which is a set of three qupit gates, $\{S, H, CX\}$, that generates all Clifford unitaries~\cite{GottesmanD1999ftqudit}.
On $Z$-basis states $\ket{k}$, the qupit $S$ gate acts as
\begin{equation}
    \ket{k} \mapsto \omega^{k(k\minu 1)/2} \ket{k} ,
\end{equation}
the $H$ gate acts as
\begin{equation}
    \ket{k} \mapsto \frac{1}{\sqrt{d}} \sum_{j=0}^{d\minu 1} \omega^{jk} \ket{j} ,
\end{equation}
and the $CX$ gate performs addition modulo $d$, acting as
\begin{equation}
    \ket{j}\ket{k} \mapsto \ket{j}\ket{(j+k)\,\,\text{mod}\,\,d}.
\end{equation}

A more helpful view of the $H$ gate is that it sends each $X$-basis state to the $Z$-basis state with the same eigenvalue, hence it is also referred to as the discrete quantum Fourier transform.

To distinguish it from non-Clifford controlled $X$ gates, we adopt the shorthand of labelling the control of the (Clifford) $CX$ gate by $\Lambda$~\cite{BocharovA2017ternaryshor} and the target by $X$ or by $+1$.
\begin{definition}\label{def:controlled-gates}
    Given a qudit unitary $U$ we define
\begin{equation}
     \Lambda(U)\ket{c}\ket{t} = \ket{c} \otimes (U^c \ket{t}),
\end{equation}
i.e.~we apply the unitary $U$ a number of times equal to to the value of the control qudit.
\end{definition}
By repeating the $CX$ gate $k \in \mathbb{Z}_d$ times, the $\Lambda(+k) \coloneqq \Lambda(X^k)$ is Clifford.

\subsection{Diagonal gates in the Clifford hierarchy}
\label{sec:cliffordplust}
Not all non-Clifford gates possess the same viability in a fault-tolerant setting.  The Clifford hierarchy is an important concept in characterizing which gates admit magic state injection protocols or transversal implementation on stabilizer codes.

\begin{definition}
The $n^{\text{th}}$ level of the Clifford hierarchy $C_n$~\cite{GottesmanD1999ftqudit,PllahaT2020unweylingclifford,RengaswamyN2019UnifyingRings,CuiS2017Diagonalhierarchy} is defined inductively as $C_1 := \mathcal{P}$, the Pauli's and
\begin{equation}
    C_n := \{\text{unitary}~U~|~\forall P\in \mathcal{P}: UPU^\dagger \in C_{n\minu 1}\}.
\end{equation}
\end{definition}

The Pauli gates form $C_1$ and the Clifford gates form $C_2$, while certain non-Clifford gates belong to $C_n$ for higher $n$. 
The Clifford hierarchy is significant because of its relation to the correction of Pauli errors.
For instance, when implementing a gate via a `state injection' mechanism~\cite{BravyiS2005ftqc}, a gate from $C_i$ requires a correction from $C_{i\minu 1}$. 
Thus, implementing a gate from the third level of the hierarchy using state injection requires a Clifford correction. This operation is cheap in most fault-tolerant architectures.

Despite the third level of the Clifford hierarchy and higher not being groups (for instance they are not closed under composition), the Paulis, Cliffords, and the diagonal gates in the Clifford hierarchy are groups.  This facilitated characterization of all diagonal gates in the Clifford hierarchy for arbitrary prime qudit dimension $d$ by Cui, Gottesman, and Krishna~\cite{CuiS2017Diagonalhierarchy}.  In particular, they identified that gates of the form
\begin{equation}\label{eq:Uma}
U_{m,a} = \sum_{j=0}^{d\minu 1}\,\,\text{exp}\left(\frac{2 \pi i}{d^m} j^a \right) \ket{j} \bra{j}
\end{equation}
lie in the $\left((d\minu 1)(m\minu 1)+a\right)^{\text{th}}$ level of the Clifford hierarchy~\cite[Equation 16]{CuiS2017Diagonalhierarchy}.

All single-qudit operations in the third level of the Clifford hierarchy admit robust state injection protocols, as commutation with a Pauli gate requires a Clifford correction.  The most well-studied single-qubit non-Clifford gate, the $T$ gate $\text{diag}(1,\text{exp}(\frac{2\pi i}{8}))$, rests in the third level of the Clifford hierarchy, and admits fault-tolerant implementation, such as through magic state distillation and state injection, or through transversal implementation on certain stabilizer codes.  The $T$ gate has been generalized to prime dimensions by means of explicitly solving for the diagonal single-qudit gates in the third level of the Clifford hierarchy~\cite{HowardM2012quditTgate}.  This exact solution involved the multiplicative inverse of 12 modulo $d$, and thus yielded the same form for all prime dimensions with the exceptions of $d = 2$ and $d = 3$.

The qutrit $T$ gate is defined as $\text{diag}(1, \omega^{1/3}, \omega^{\minu 1/3})$~\cite{HowardM2012quditTgate}, or a diagonal gate which is Clifford equivalent~\cite{BocharovA2017ternaryshor}.  Like the qubit $T$ gate, the qutrit $T$ gate belongs to the third level of the Clifford hierarchy, can be injected into a circuit using magic states, and its magic states can be distilled by magic state distillation~\cite{CampbellE2012tgatedistillation}. This means that we can fault-tolerantly implement this gate on many qutrit error correcting codes.
Also, as is for qubits, the qutrit Clifford+$T$ gate set is approximately universal, meaning that we can approximate any qutrit unitary using just Clifford gates and the $T$ gate~\cite[Theorem 1]{CuiS2015universalmetaplectic}.  More generally, adding a single-qupit non-Clifford gate to the Clifford gate set yields approximately universal quantum computation.

The $T$ gate is special for $d = 2$ and for $d = 3$.  For all other prime dimensions, exactly solving for the single-qudit diagonal gates in the third level of the Clifford hierarchy yields a family of gates all satisfying the same explicit expression~\cite{HowardM2012quditTgate}, of which a sensible choice is $U_{1,3}$ in Equation~\eqref{eq:Uma}~\cite{KrishnaA2019lowoverheadmsd}.

In this text we investigate a different class of gates, by a different generalization of the qutrit $T$ gate to all odd prime dimensions.  The primary reason is to explore non-Clifford gates which enable potentially fault-tolerant synthesis of important logical gates and resource states.  We apply these gates to constructing qubit and qupit generalized $W$ states, for which, for most sizes of the $W$ state, no construction which is both deterministic and potentially fault-tolerant was previously known.

We now introduce the sole single-qudit non-Clifford gate we will need for our constructions.
To the best of our knowledge, this gate has not been studied in the context of gate synthesis before for $d > 3$.
Let $\zeta \coloneqq e^{\frac{2\pi i}{d^2}} = \sqrt[d]{\omega}$ denote the ${d^2}^{\text{th}}$ root of unity.  We denote $\sqrt[d]{Z} = U_{2,1}$ in Equation~\eqref{eq:Uma} to be the gate acting as
\begin{equation}
\ket{k} \mapsto \zeta^k \ket{k}
\end{equation}
for $k$ in $\mathbb{Z}_d$.
Indeed $(\sqrt[d]{Z} \,)^d = Z$, justifying it being a $d^{\text{th}}$ root of $Z$ gate, albeit not the only gate to satisfy this condition.

According to Equation~\eqref{eq:Uma}, the $\sqrt[d]{Z}$ gate is in the $d^{\text{th}}$ level of the Clifford hierarchy.  This implies that for qubits it is the Clifford $S = \text{diag}(1, i)$ gate.  For qutrits it is $\text{diag}(1, \omega^{1/3}, \omega^{2/3}) = U_{2,1}$ in Equation~\eqref{eq:Uma}, joining the other definitions of the qutrit $T$ gate in the literature which are all Clifford equivalent to one another.

For all other $d$, this gate is above the third level of the Clifford hierarchy.  This means that whether is can be efficiently implemented fault-tolerantly is an open question.  It was conjectured by de Silva informed by numerical evidence that single-qupit gates in any level of the Clifford hierarchy can be implemented efficiently~\cite{deSilvaN2021qupitgateteleportation}; to date, this has been proven to be the case for every two-qubit gate in any level of the Clifford hierarchy~\cite{ZengB2008semicliffordtwoqubit}, for every third level gate of two qutrits~\cite{deSilvaN2021qupitgateteleportation}, and for every third level gate of one qudit (of any prime dimension)~\cite{deSilvaN2021qupitgateteleportation}.

Another gate we will refer to is the single-qudit gate $P_1(k)$:
\begin{definition}~\label{def:Pmk}
    In the $(d\minu 1)^{\text{th}}$ level of the Clifford hierarchy~\cite{CuiS2017Diagonalhierarchy}:
\begin{equation}
    P_1(k) \coloneqq \sum_{j=0,\\ j \neq k}^{d\minu 1} \ket{j}\bra{j} + \,\,\omega \ket{k}\bra{k}
\end{equation}
\end{definition}

\section{Controlled gates and \texorpdfstring{$W$}{W} states}\label{sec:controlledandW}
In addition to the $CX$ gate, we will need two more controlled gates in this text: the $\ket{0}$-controlled $X$ gate and the controlled $H$ gate.

\subsection{The \texorpdfstring{$\ket{0}$}{0}-controlled \texorpdfstring{$X$}{X} gate}\label{sec:ZCX}
The $\ket{0}$-controlled $X$ gate acts on the target as
\begin{equation}
    \begin{cases}
        X \ &\text{if control is}\ \ket{0} \\
        I \ &\text{if control is}\ \ket{k}, \ k \neq 0 \in \mathbb{Z}_d
    \end{cases}
\end{equation}
Notably, it is non-Clifford except in the qubit case.  We present its construction below.  Before we do so, we define what is a \emph{controlled global phase}~\cite{YehL2022controlunivqutrit}, before presenting an instance of a circuit we will soon apply to correct a specific controlled global phase.

A complication when trying to construct controlled unitaries, is that usually irrelevant global phases becomes `local' and hence must be dealt with accordingly~\cite[Lemma 5.2]{BarencoA1995elementarygates}.
\begin{definition}\label{def:controlledglobalphasegate}
     A \emph{controlled global phase gate} is a controlled unitary where the unitary is $e^{i \phi} \mathbb{I}$, for identity matrix $I$ and phase $\phi$.
 \end{definition}
The number of qudits the identity matrix acts on in this definition is irrelevant as the phase factor can be ``factored out'' from the tensor product of the control and target qudits:
 \begin{equation}
     \label{eq:controlledglobalphase}
     \begin{tikzpicture}
	\begin{pgfonlayer}{nodelayer}
		\node [style=0 control] (0) at (-4.75, -0.5) {\textsf{0}};
		\node [style=none] (1) at (-4.75, 2) {};
		\node [style=none] (2) at (-4.75, 0.25) {};
		\node [style=none] (3) at (-4.75, 1.375) {$\vdots$};
		\node [style=0 control] (4) at (-4.75, 3) {\textsf{0}};
		\node [style=box] (5) at (-4.75, -2.5) {$\omega^{\gamma} I$};
		\node [style=none] (6) at (-7, -2.5) {};
		\node [style=none] (7) at (-7, 3) {};
		\node [style=none] (8) at (-7, -0.5) {};
		\node [style=none] (9) at (-2.5, 3) {};
		\node [style=none] (10) at (-2.5, -0.5) {};
		\node [style=none] (11) at (-2.5, -2.5) {};
		\node [style=none] (12) at (-3.25, -2.625) {};
		\node [style=none] (13) at (-3, -2.375) {};
		\node [style=none] (15) at (-6.5, -2.625) {};
		\node [style=none] (16) at (-6.25, -2.375) {};
		\node [style=none] (17) at (-6.5, -2.75) {\tiny$m$};
		\node [style=none] (18) at (-7.5, 3.25) {};
		\node [style=none] (19) at (-7.5, -1.75) {};
		\node [style=none] (20) at (-8, 0.75) {};
		\node [style=none] (22) at (-8.375, 0.75) {$n$};
		\node [style=0 control] (23) at (-4.75, -1.5) {\textsf{0}};
		\node [style=none] (25) at (-7, -1.5) {};
		\node [style=none] (26) at (-2.5, -1.5) {};
		\node [style=0 control] (27) at (1, -0.5) {\textsf{0}};
		\node [style=none] (28) at (1, 2) {};
		\node [style=none] (29) at (1, 0.25) {};
		\node [style=none] (30) at (1, 1.375) {$\vdots$};
		\node [style=0 control] (31) at (1, 3) {\textsf{0}};
		\node [style=none] (33) at (-0.5, -2.5) {};
		\node [style=none] (34) at (-0.5, 3) {};
		\node [style=none] (35) at (-0.5, -0.5) {};
		\node [style=none] (36) at (2.5, 3) {};
		\node [style=none] (37) at (2.5, -0.5) {};
		\node [style=none] (38) at (2.5, -2.5) {};
		\node [style=none] (42) at (0.875, -2.625) {};
		\node [style=none] (43) at (1.125, -2.375) {};
		\node [style=box] (50) at (1, -1.5) {$P_1(0)^{\gamma}$};
		\node [style=none] (51) at (-0.5, -1.5) {};
		\node [style=none] (52) at (2.5, -1.5) {};
		\node [style=none] (53) at (-1.5, 0) {=};
		\node [style=0 control] (54) at (6, 1.75) {\textsf{0}};
		\node [style=none] (55) at (6, 2.65) {};
		\node [style=none] (56) at (6, 2.1) {};
		\node [style=0 control] (58) at (6, 3) {\textsf{0}};
		\node [style=none] (59) at (4.5, -2.5) {};
		\node [style=none] (60) at (4.5, 3) {};
		\node [style=none] (61) at (4.5, 1.75) {};
		\node [style=none] (62) at (7.5, 3) {};
		\node [style=none] (63) at (7.5, 1.75) {};
		\node [style=none] (64) at (7.5, -2.5) {};
		\node [style=none] (65) at (5.875, -2.625) {};
		\node [style=none] (66) at (6.125, -2.375) {};
		\node [style=box] (73) at (6, 0.75) {$P_1(0)^{\gamma}$};
		\node [style=none] (74) at (4.5, 0.75) {};
		\node [style=none] (75) at (7.5, 0.75) {};
		\node [style=none] (76) at (3.5, 0) {=};
		\node [style=0 control] (77) at (6, -1.5) {\textsf{0}};
		\node [style=0 control] (81) at (6, -0.25) {\textsf{0}};
		\node [style=none] (82) at (4.5, -0.25) {};
		\node [style=none] (83) at (4.5, -1.5) {};
		\node [style=none] (84) at (7.5, -0.25) {};
		\node [style=none] (85) at (7.5, -1.5) {};
		\node [style=none] (86) at (6, -0.6) {};
		\node [style=none] (87) at (6, -1.15) {};
		\node [style=none] (89) at (6, 2.375) {$\rotatebox{90}{...}$};
		\node [style=none] (90) at (6, -0.875) {$\rotatebox{90}{...}$};
		\node [style=none] (91) at (-3.25, -2.75) {\tiny$m$};
		\node [style=none] (92) at (0.875, -2.75) {\tiny$m$};
		\node [style=none] (93) at (5.875, -2.75) {\tiny$m$};
	\end{pgfonlayer}
	\begin{pgfonlayer}{edgelayer}
		\draw (2.center) to (0);
		\draw (1.center) to (4);
		\draw (6.center) to (5);
		\draw (7.center) to (4);
		\draw (8.center) to (0);
		\draw (5) to (11.center);
		\draw (0) to (10.center);
		\draw (4) to (9.center);
		\draw (12.center) to (13.center);
		\draw (15.center) to (16.center);
		\draw [in=180, out=0, looseness=0.75] (20.center) to (19.center);
		\draw [in=-180, out=0, looseness=0.75] (20.center) to (18.center);
		\draw (25.center) to (23);
		\draw (23) to (26.center);
		\draw (0) to (23);
		\draw (23) to (5);
		\draw (29.center) to (27);
		\draw (28.center) to (31);
		\draw (34.center) to (31);
		\draw (35.center) to (27);
		\draw (27) to (37.center);
		\draw (31) to (36.center);
		\draw (42.center) to (43.center);
		\draw [in=180, out=0] (51.center) to (50);
		\draw (50) to (52.center);
		\draw (27) to (50);
		\draw (33.center) to (38.center);
		\draw (56.center) to (54);
		\draw (55.center) to (58);
		\draw (60.center) to (58);
		\draw (61.center) to (54);
		\draw (54) to (63.center);
		\draw (58) to (62.center);
		\draw (65.center) to (66.center);
		\draw (74.center) to (73);
		\draw (73) to (75.center);
		\draw (54) to (73);
		\draw (59.center) to (64.center);
		\draw (82.center) to (81);
		\draw (83.center) to (77);
		\draw (77) to (85.center);
		\draw (81) to (84.center);
		\draw (73) to (81);
		\draw [style=simple] (81) to (86.center);
		\draw [style=simple] (87.center) to (77);
	\end{pgfonlayer}
\end{tikzpicture}
 \end{equation}
Here we substituted the global phase by $\gamma$ where $\phi = \gamma \cdot \frac{2\pi}{d}$ so that the phases are powers of $\omega$.

Our circuit decomposition of the $\ket{0}$-controlled $X$ gate will be off by a certain controlled global phase, which we later calculate in Equation~\eqref{eq:calccgp} to be $\omega^{(d\minu 1)/2}$.
If it is desired to correct this controlled global phase, we want to implement a diagonal gate $CGP$ which (up to a global phase we can ignore) acts as
\begin{equation}
    \ket{0} \mapsto \omega^{(d\minu 1)/2} \ket{0}, \\
    \ket{k} \mapsto \ket{k}, \ k \neq 0 \in \mathbb{Z}_d.
\end{equation}

\begin{lemma}\label{lemma:cgp}
    For any odd prime qudit dimension $d$, a two-qudit circuit can be constructed to correct a controlled global phase of $\omega^{(d\minu 1)/2}$ on the target when the control is $\ket{0}$, with $(d\minu 1)$ $\sqrt[d]{Z}$ gates.
\end{lemma}
\begin{proof}
    We implement this by applying $(d\minu 1)/2$ times the gate $P_1(k)$ from Definition~\ref{def:Pmk} where $k = 0$.  Up to a global phase of $\zeta$, for $k \in \mathbb{Z}_d$, the $P_1(k)$ gates are also exactly synthesizable in the Clifford+$\sqrt[d]{Z}$ gate set by conjugating $\sqrt[d]{Z}$ by Pauli $X$:
\begin{equation}\label{eq:Pk}
    P_1(k) = \zeta \,\, X^k \, {\sqrt[d]{Z}}^{\dagger} \, X \, \sqrt[d]{Z} \, {X^{\dagger}}^{(k+1)}
\end{equation}
This follows from direct computation for the $k = 0$ case, as $X \, \sqrt[d]{Z} \, X^{\dagger} = \text{diag}(\zeta^{d\minu 1},1,\zeta,\zeta^2,...,\zeta^{d\minu 2})$, and ${\sqrt[d]{Z}}^{\dagger} = \text{diag}(1,\zeta^{\minu 1},\zeta^{\minu 2},...,\zeta^{\minu (d\minu 1)})$. For all other $k \in \mathbb{Z}_d$, conjugate by $X$ $k$ times.
\end{proof}
\begin{remark}
    As $P_1(0)$ is one level lower in the Clifford hierarchy than $\sqrt[d]{Z}$, it could be sensible to add it to the Clifford+$\sqrt[d]{Z}$ gate set, instead of decomposing it into $\sqrt[d]{Z}$ gates.
\end{remark}

The following proposition essentially generalizes the construction for the $\ket{2}$-controlled qutrit $Z$ gate by Bocharov, Cui, Roetteler, and Svore~\cite{BocharovA2016ternaryarithmetics} to arbitrary prime qudit dimension.
\begin{proposition}\label{prop:zcx}
The $\ket{0}$-controlled $X$ gate can be constructed with $\sqrt[d]{Z}$-count $d$ (up to a controlled global phase that is a power of $\omega$), or with $\sqrt[d]{Z}$-count $2d\minu 1$ (exact implementation).
\end{proposition}
\begin{proof}
\begin{equation}\label{eq:ZCX}
    \begin{tikzpicture}
	\begin{pgfonlayer}{nodelayer}
		\node [style=none] (240) at (-3.125, 1.25) {};
		\node [style=none] (296) at (7.625, -0.75) {};
		\node [style=none] (300) at (7.625, 1.25) {};
		\node [style=none] (305) at (-3.125, -0.75) {};
		\node [style=box] (311) at (-1.625, -0.75) {$H$};
		\node [style=box] (322) at (2.625, -0.75) {$\sqrt[d]{Z}$};
		\node [style=0 control] (323) at (0.875, 1.25) {$\Lambda$};
		\node [style=box] (324) at (0.875, -0.75) {$X$};
		\node [style=none] (339) at (0, -1.75) {};
		\node [style=none] (340) at (3.5, -1.75) {};
		\node [style=none] (341) at (1.75, -2) {};
		\node [style=none] (342) at (1.75, -2.5) {\small{repeat $d$ times}};
		\node [style=box] (343) at (5.375, -0.75) {$H^{\dagger}$};
		\node [style=box] (344) at (5.625, 1.25) {$CGP$};
		\node [style=0 control] (345) at (-7.125, 1.25) {\textsf{0}};
		\node [style=none] (346) at (-8.5, 1.25) {};
		\node [style=none] (347) at (-5.75, 1.25) {};
		\node [style=box] (348) at (-7.125, -0.75) {$X$};
		\node [style=none] (349) at (-8.5, -0.75) {};
		\node [style=none] (350) at (-5.75, -0.75) {};
		\node [style=none] (351) at (-4.5, 0.25) {$=$};
	\end{pgfonlayer}
	\begin{pgfonlayer}{edgelayer}
		\draw [style=simple] (324) to (323);
		\draw [style=simple, in=120, out=-90, looseness=0.75] (339.center) to (341.center);
		\draw [style=simple, in=-90, out=60, looseness=0.75] (341.center) to (340.center);
		\draw [style=simple] (305.center) to (296.center);
		\draw [style=simple] (300.center) to (240.center);
		\draw [style=simple] (347.center) to (346.center);
		\draw [style=simple] (349.center) to (350.center);
		\draw [style=simple] (348) to (345);
	\end{pgfonlayer}
\end{tikzpicture}
\end{equation}
To derive that the above circuit implements the $\ket{0}$-controlled $X$ gate, start by considering only the action of the subcircuit which is repeated $d$ times.  Observe that when the control and target are $Z$-basis states $\ket{x}$ and $\ket{y}$ for $x, y \in \mathbb{Z}_d$, at the time slice after the first $CX$ gate the state is $\ket{x, (x+y)\,\,\text{mod}\,\,d}$.  The $\sqrt[d]{Z}$ gate on the target then contributes a factor of $\zeta^{(x + y)\,\,\text{mod}\,\,d}$.  The $\sqrt[d]{Z}$ gate after that contributes a factor of $\zeta^{(2x + y)\,\,\text{mod}\,\,d}$, and so on.  Repetition of the $CX$ and the $\sqrt[d]{Z}$ gate $d$ times computes to the diagonal gate
\begin{equation}\label{eq:zcx}
\ket{x, y} \mapsto \prod_{j=0}^{d\minu 1} \zeta^{(jx + y)\,\,\text{mod}\,\,d} \ket{x, y}.
\end{equation}

This can be broken down into two cases: $x= 0$ and $x \neq 0$.
When $x = 0$, the output of Equation~\eqref{eq:zcx} is
\begin{equation}
\zeta^{dy} \ket{x, y}.
\end{equation}
As $\zeta^d = \omega$, this is simply the $Z$ gate on the target qupit (except for qubits where this is the $S$ gate).

When $x \neq 0$, the output is
\begin{equation}
\zeta^{\sum_{j=0}^{d\minu 1} (jx + y)\,\,\text{mod}\,\,d} \ket{x, y}.
\end{equation}
Note that in the sum, $(jx + y)\,\,\text{mod}\,\,d$ takes on the value of each element in $\mathbb{Z}_d$ exactly once.  Therefore, when $x \neq 0$, the output is
\begin{equation}\label{eq:calccgp}
\zeta^{\sum_{j=0}^{d\minu 1} j} \ket{x, y} = \zeta^{d(d\minu 1)/2} \ket{x, y}.
\end{equation}
As $\zeta^d = \omega$, this amounts to a controlled global phase of $\omega^{(d\minu 1)/2}$ which applies iff $x \neq 0$, labelled $CGP$ in Equation~\eqref{eq:ZCX}.  This can be either corrected by applying Lemma~\ref{lemma:cgp}, or omitted in situations where this phase discrepancy is tolerated to improve efficiency. 

Finally, we conjugate the target qudit by $H$ to obtain the $\ket{0}$-controlled $X$ gate (up to a global phase that we can safely ignore).
\end{proof}
For $d=3$ $CGP$ is Clifford, so the $\ket{0}$-controlled $X$ gate has $T$-count 3~\cite{BocharovA2016ternaryarithmetics}.
For $d=2$, while $(d\minu 1)/2$ is not an integer, the construction remains valid by substituting $CGP = Z$.

\begin{corollary}
For $k, p \in \mathbb{Z}_d$, conjugating the control qudit by $X^k$ yields the $\ket{k}$-controlled $X$ gate.  From this, the $\ket{k}$-controlled $X^p$ gate can be built na\"ively by repeating the construction $p$ times.
\end{corollary}

\subsection{The \texorpdfstring{$d$}{d}-qubit \texorpdfstring{$W$}{W} state}\label{sec:Wdstate}
We now introduce our first $W$ state construction.
\begin{proposition}\label{prop:dqubitWstate}
For any prime qudit dimension $d$, the $N=d$-qubit $W$ state can be constructed deterministically in the Clifford+$\sqrt[d]{Z}$ gate set using $d^2 \minu d$ $\sqrt[d]{Z}$ gates.
\end{proposition}
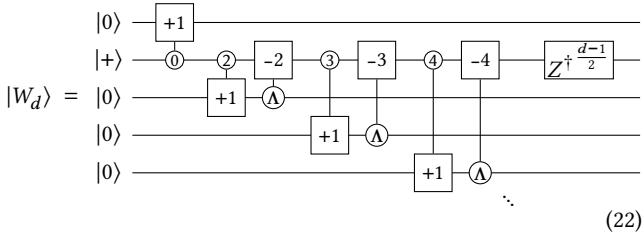
\begin{proof}

\begin{equation}\label{eq:Wcirc}
    \begin{tikzpicture}
	\begin{pgfonlayer}{nodelayer}
		\node [style=none] (183) at (-8.125, 0) {$=$};
		\node [style=none] (240) at (-6.5, 1) {};
		\node [style=none] (241) at (-6.5, 0) {};
		\node [style=none] (246) at (-6.5, 2) {};
		\node [style=none] (295) at (7, 0) {};
		\node [style=none] (296) at (7, -1) {};
		\node [style=none] (297) at (7, -2) {};
		\node [style=none] (299) at (7, 2) {};
		\node [style=none] (300) at (7, 1) {};
		\node [style=none] (305) at (-6.5, -1) {};
		\node [style=none] (306) at (-6.5, -2) {};
		\node [style=none] (309) at (-9.25, 0) {$|W_d \rangle$};
		\node [style=0 control] (314) at (-5.375, 1) {\textsf{0}};
		\node [style=box] (315) at (-5.375, 2) {$+1$};
		\node [style=0 control] (317) at (-4, 1) {\textsf{2}};
		\node [style=box] (318) at (-4, 0) {$+1$};
		\node [style=0 control] (319) at (-2.75, 0) {$\Lambda$};
		\node [style=box] (320) at (-2.75, 1) {$\minu 2$};
		\node [style=0 control] (321) at (-1.25, 1) {\textsf{3}};
		\node [style=box] (322) at (-1.25, -1) {$+1$};
		\node [style=0 control] (323) at (0, -1) {$\Lambda$};
		\node [style=box] (324) at (0, 1) {$\minu 3$};
		\node [style=0 control] (325) at (1.5, 1) {\textsf{4}};
		\node [style=box] (326) at (1.5, -2) {$+1$};
		\node [style=0 control] (327) at (2.75, -2) {$\Lambda$};
		\node [style=box] (328) at (2.75, 1) {$\minu 4$};
		\node [style=none] (329) at (3.5, -2.75) {$\rotatebox{135}{...}$};
		\node [style=none] (330) at (-7.125, 2) {$|0\rangle$};
		\node [style=none] (331) at (-7.125, 1) {$|+\rangle$};
		\node [style=none] (332) at (-7.125, 0) {$|0\rangle$};
		\node [style=none] (333) at (-7.125, -1) {$|0\rangle$};
		\node [style=none] (334) at (-7.125, -2) {$|0\rangle$};
		\node [style=box] (335) at (5.375, 1) {\large$\,\,\,\,\,\,\,\,\,\,\,\,$};
		\node [style=none] (336) at (5.375, 1) {${Z^{\dagger}}^{\frac{d\minu 1}{2}}$};
	\end{pgfonlayer}
	\begin{pgfonlayer}{edgelayer}
		\draw [style=simple] (315) to (314);
		\draw [style=simple] (318) to (317);
		\draw [style=simple] (320) to (319);
		\draw [style=simple] (322) to (321);
		\draw [style=simple] (324) to (323);
		\draw [style=simple] (326) to (325);
		\draw [style=simple] (328) to (327);
		\draw [style=simple] (246.center) to (299.center);
		\draw [style=simple, in=360, out=180] (300.center) to (240.center);
		\draw [style=simple] (241.center) to (295.center);
		\draw [style=simple] (296.center) to (305.center);
		\draw [style=simple] (306.center) to (297.center);
	\end{pgfonlayer}
\end{tikzpicture}
\end{equation}
Correctness of the circuit is straightforward to verify by observing that the $d$ $Z$-basis states which make up the $\ket{+}$ state in equal superposition each contribute to one of the $d$ terms of the output $W$ state.
Consider the unitary circuit being applied to the state
\begin{equation}
\ket{0}\ket{+}\ket{0}^{\otimes (d\minu 2)} = \frac{1}{\sqrt{d}} \sum_{j=0}^{d\minu 1} \ket{0}\ket{j}\ket{0}^{\otimes (d\minu 2)}.
\end{equation}
Observe that each subcircuit consisting of the $\ket{k}$-controlled $X$ gate followed by $k$ ${CX}^{\dagger}$ gates, for $k \geq 2 \in \mathbb{Z}_d$, acts as
\begin{equation}
\ket{0}\ket{k}\ket{0}^{\otimes (d\minu 2)} \mapsto \ket{0}^{\otimes k} \ket{1} \ket{0}^{\otimes (d\minu k\minu 1)}
\end{equation}
while being identity for all $j \neq k$ on $\ket{0}\ket{j}\ket{0}^{\otimes (d\minu 2)}$ and on\linebreak $\ket{0}^{\otimes j} \ket{1} \ket{0}^{\otimes (d\minu j\minu 1)}$.  For $d(d\minu 1)$ $\sqrt[d]{Z}$ count, delay all controlled global phase corrections until applying $Z^{\dagger}$ $\frac{d\minu 1}{2}$ times at the end.
\end{proof}

\begin{corollary}
As the qutrit $CGP$ is Clifford, the three-qubit $W$ state can be implemented with qutrit $T$-count 3:
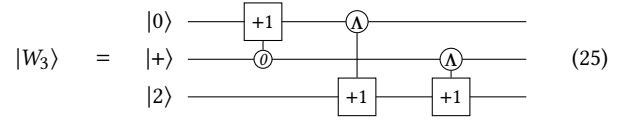
\begin{equation}
    \begin{tikzpicture}
	\begin{pgfonlayer}{nodelayer}
		\node [style=none] (183) at (-5.5, 0) {$=$};
		\node [style=none] (240) at (-3.25, 0) {};
		\node [style=none] (241) at (-3.25, -1) {};
		\node [style=none] (246) at (-3.25, 1) {};
		\node [style=none] (295) at (5.75, -1) {};
		\node [style=none] (299) at (5.75, 1) {};
		\node [style=none] (300) at (5.75, 0) {};
		\node [style=none] (309) at (-7.25, 0) {$|W_3 \rangle$};
		\node [style=0 control] (314) at (-1.25, 0) {$\textsf{0}$};
		\node [style=box] (315) at (-1.25, 1) {$+1$};
		\node [style=box] (318) at (1.25, -1) {$+1$};
		\node [style=none] (330) at (-4, 1) {$|0\rangle$};
		\node [style=none] (331) at (-4, 0) {$|+\rangle$};
		\node [style=none] (332) at (-4, -1) {$|2\rangle$};
		\node [style=0 control] (333) at (1.25, 1) {$\Lambda$};
		\node [style=box] (336) at (3.75, -1) {$+1$};
		\node [style=0 control] (337) at (3.75, 0) {$\Lambda$};
	\end{pgfonlayer}
	\begin{pgfonlayer}{edgelayer}
		\draw [style=simple, in=90, out=-90] (315) to (314);
		\draw [style=simple] (246.center) to (299.center);
		\draw [style=simple] (300.center) to (240.center);
		\draw [style=simple] (241.center) to (295.center);
		\draw [style=simple] (333) to (318);
		\draw [style=simple] (337) to (336);
	\end{pgfonlayer}
\end{tikzpicture}
\end{equation}
\end{corollary}
This achieves a new minimum single-qupit non-Clifford gate count to synthesize the $W$ state, and moreover deterministically.

\begin{remark}
When $d=2$, the construction in Proposition~\ref{prop:dqubitWstate} does indeed produce $\frac{1}{\sqrt{2}} \left( \ket{01}+\ket{10} \right)$.  However, as this is equivalent to the Bell state up to an $X$ gate, this does not possess the usefulness of a $W$ state of three or more qubits.
\end{remark}

\subsection{The controlled qupit \texorpdfstring{$H$}{H} gate}\label{sec:CH}
In the previous section, we showed how to construct the generalized qubit $W$ state of any prime size.  However, the scaling of $O(N^2)$ $\sqrt[d]{Z}$ gate count implies that this approach is less asymptotically efficient for large $d$.  In this section, in order to improve this to asymptotically linear gate count, we present how to construct a controlled qupit $H$ gate, up to a controlled global phase not of practical consequence in this work.  We do this by extending our controlled $H$ for qutrits in Ref.~\cite{YehL2022controlunivqutrit} to any prime qudit dimension.
\begin{lemma}
    Up to a global phase, in any prime qudit dimension, the $H$ gate can be decomposed into Clifford phase gates.
\end{lemma}
\begin{proof}
    The qudit $H$ gate can be decomposed into a $Z$ rotation followed by an $X$ rotation followed by a $Z$ rotation as~\cite{WangQ2021qufinite,NgKF2019euqpltalk}:
    \begin{align}\label{eq:EU}
        H &= i^{(d\minu 1)/2} \quad Z(\vec{\tau}) X(\vec{\tau}) Z(\vec{\tau})\\
        &= i^{(d\minu 1)/2} \,\, S \, Z^{(2^{\minu 1}\,\text{mod}\,d)} \,\, H^{\dagger} \, S \, Z^{(2^{\minu 1}\,\text{mod}\,d)} \, H \,\, S \, Z^{(2^{\minu 1}\,\text{mod}\,d)}
    \end{align}
    where
    \begin{equation}
        \tau(k) \coloneqq \omega^{2^{\minu 1} k^2\,\,\text{mod}\,\,d},
    \end{equation}
    \begin{equation}
        Z(\vec{v}) \coloneqq \sum_{k=0}^{d\minu 1} \omega^{v(k)} \ket{k}\bra{k},
    \end{equation}
    and
    \begin{equation}
        X(\vec{v}) = H^{\dagger} Z(\vec{v}) H.
    \end{equation}
\end{proof}

Our approach to controlling the $H$ gate is to control each component $Z$ and $X$ phase gate in Equation~\eqref{eq:EU}.  We apply strategies from past work: generalizing controlled phase gates from qubits to qutrits~\cite{vandeWetering2022qutritphasegadgets} and controlling the qutrit $H$ gate~\cite{GlaudellA2022qutritmetaplecticsubset}.  In future work we will consider generalizing to higher dimensions, adding any number of controls to qutrit gates~\cite{YehL2022controlunivqutrit}.

The key structural insight leverages the generalization to arbitrary dimensions of \emph{phase gadgets}, a symmetric multi-qudit interaction which applies a phase dependent on the sum of all computational basis states modulo $d$.  Phase gadgets have been characterized thus far in two graphical calculi: in the qubit ZX-calculus~\cite{CoeckeB2011interacting} in Ref.~\cite{KissingerA2020reduc,CowtanA2020phasegadget}, and in the qutrit ZX-calculus~\cite{WangQ2014qutritcalculus} in Ref.~\cite{vandeWetering2022qutritphasegadgets}.  We remark that phase gadgets are broadly applicable to all qudit dimensions as a graphical construct to relate diagonal gates in the circuit model and polynomial functions.

Hence, this gives us an idea of what sort of circuit structure may be a good place to start.  Keeping minimal on introducing formalism of the qudit $ZX$-calculus~\cite{RanchinA2014quditzx,WangQ2021qufinite,BoothR2022qupitstab}, we introduce non-unitary circuit elements which respectively act as addition modulo $d$ and copying of basis states, labelled $\ket{x}$ and $\ket{y}$ below:
\begin{equation}\label{eq:copysum}
    \begin{tikzpicture}
	\begin{pgfonlayer}{nodelayer}
		\node [style=none] (0) at (-7.75, 0.5) {};
		\node [style=none] (1) at (-7.75, -0.5) {};
		\node [style=X dot] (2) at (-7.25, 0) {};
		\node [style=none] (3) at (-6.5, 0) {};
		\node [style=none] (4) at (-8.25, 0.5) {$|x\rangle$};
		\node [style=none] (5) at (-8.25, -0.5) {$|y\rangle$};
		\node [style=none] (6) at (-6, 0) {$=$};
		\node [style=none] (7) at (-3.25, 0) {$|(x+y) \textnormal{ mod } d\rangle$};
		\node [style=none] (8) at (3.75, 0.5) {};
		\node [style=none] (9) at (3.75, -0.5) {};
		\node [style=Z dot] (10) at (3.25, 0) {};
		\node [style=none] (11) at (2.75, 0) {};
		\node [style=none] (12) at (2.25, 0) {$|x\rangle$};
		\node [style=none] (13) at (4.25, 0) {$=$};
		\node [style=none] (14) at (5.25, 0.5) {$|x\rangle$};
		\node [style=none] (15) at (5.25, -0.5) {$|x\rangle$};
		\node [style=none] (16) at (-1.125, 0) {};
		\node [style=none] (17) at (-0.125, 0) {};
		\node [style=none] (18) at (5.75, 0.5) {};
		\node [style=none] (19) at (6.75, 0.5) {};
		\node [style=none] (20) at (5.75, -0.5) {};
		\node [style=none] (21) at (6.75, -0.5) {};
	\end{pgfonlayer}
	\begin{pgfonlayer}{edgelayer}
		\draw [style=simple] (0.center) to (2);
		\draw [style=simple] (3.center) to (2);
		\draw [style=simple] (2) to (1.center);
		\draw [style=simple] (8.center) to (10);
		\draw [style=simple] (11.center) to (10);
		\draw [style=simple] (10) to (9.center);
		\draw (17.center) to (16.center);
		\draw (19.center) to (18.center);
		\draw (21.center) to (20.center);
	\end{pgfonlayer}
\end{tikzpicture}
\end{equation}

More generally, these are called spiders.  They can have any number $m$ and $n$ of input and output wires, and can be labelled by a phase vector $\vec{\alpha}$ of length $d\minu 1$.  By convention, when $\vec{\alpha}$ is the zero vector, the label is omitted.
\begin{equation}
    \begin{tikzpicture}
	\begin{pgfonlayer}{nodelayer}
		\node [style=Z phase dot] (0) at (0, 0) {$\vec{\alpha}$};
		\node [style=none] (1) at (1, 0.5) {};
		\node [style=none] (2) at (-1, 0.5) {};
		\node [style=none] (3) at (-1, -0.5) {};
		\node [style=none] (4) at (1, -0.5) {};
		\node [style=none, rotate=90] (7) at (-0.75, 0) {$...$};
		\node [style=none, rotate=90] (8) at (0.75, 0) {$...$};
	\end{pgfonlayer}
	\begin{pgfonlayer}{edgelayer}
		\draw [style=none, in=-141, out=0, looseness=0.75] (3.center) to (0);
		\draw [style=none, in=180, out=-39, looseness=0.75] (0) to (4.center);
		\draw [style=none, in=180, out=39, looseness=0.75] (0) to (1.center);
		\draw [style=none, in=141, out=0, looseness=0.75] (2.center) to (0);
	\end{pgfonlayer}
\end{tikzpicture} \ \coloneqq \ \ket{0}^{\otimes n}\bra{0}^{\otimes m} \ + \ \sum_{j=1}^{d\minu 1} {\omega}^{\alpha(j)} \ket{j}^{\otimes n}\bra{j}^{\otimes m}
\end{equation}
The X spider is defined likewise, substituting X-basis states for the Z-basis states here.  The Z and X spiders satisfy the fusion rule between same color spiders:
\begin{equation}
    \begin{tikzpicture}
	\begin{pgfonlayer}{nodelayer}
		\node [style=Z phase dot] (0) at (-1.75, -0.875) {$\beta$};
		\node [style=none] (1) at (-3.5, -0.375) {};
		\node [style=none] (2) at (-3.5, -1.375) {};
		\node [style=none] (3) at (-0.75, -0.375) {};
		\node [style=none] (4) at (-0.75, -1.375) {};
		\node [style=Z phase dot] (5) at (-2.5, 0.875) {$\alpha$};
		\node [style=none] (6) at (-0.75, 1.375) {};
		\node [style=none] (7) at (-0.75, 0.375) {};
		\node [style=none] (8) at (-3.5, 1.375) {};
		\node [style=none] (9) at (-3.5, 0.375) {};
		\node [style=none] (10) at (-3.125, 0.875) {$\rotatebox{90}{...}$};
		\node [style=none] (11) at (-2.15, 0.025) {$...$};
		\node [style=none] (12) at (-1.25, 0.875) {$\rotatebox{90}{...}$};
		\node [style=none] (13) at (-1, -0.875) {$\rotatebox{90}{...}$};
		\node [style=none] (14) at (-3, -0.875) {$\rotatebox{90}{...}$};
		\node [style=none] (15) at (0.75, 1.25) {};
		\node [style=none] (16) at (0.75, -1.25) {};
		\node [style=none] (17) at (3.75, 1.25) {};
		\node [style=none] (18) at (3.75, -1.25) {};
		\node [style=Z phase dot] (19) at (2.25, 0) {$\,\alpha + \beta\,$};
		\node [style=none] (20) at (0, 0) {$=$};
		\node [style=none] (21) at (1, 0.125) {$\vdots$};
		\node [style=none] (22) at (3.5, 0.125) {$\vdots$};
	\end{pgfonlayer}
	\begin{pgfonlayer}{edgelayer}
		\draw [bend left=15] (1.center) to (0);
		\draw [bend left=15] (0) to (2.center);
		\draw [bend right=15] (0) to (4.center);
		\draw [bend right=15] (3.center) to (0);
		\draw [bend right=15] (6.center) to (5);
		\draw [bend right=15] (5) to (7.center);
		\draw [bend left=15] (5) to (9.center);
		\draw [bend left=15] (8.center) to (5);
		\draw [bend left] (5) to (0);
		\draw [bend left] (0) to (5);
		\draw [bend left=15, looseness=0.75] (15.center) to (19);
		\draw [bend left=15, looseness=0.75] (19) to (16.center);
		\draw [bend left=15, looseness=0.75] (18.center) to (19);
		\draw [bend left=15, looseness=0.75] (19) to (17.center);
	\end{pgfonlayer}
\end{tikzpicture}
\end{equation}

The $CX$ gate is then representable as the composition of the two operations in Equation~\eqref{eq:copysum}, whereas its inverse $CX^{\dagger}$ is its diagram flipped horizontally.
\begin{equation}
    \begin{tikzpicture}
	\begin{pgfonlayer}{nodelayer}
		\node [style=none] (0) at (-3.25, 0.75) {};
		\node [style=Z dot] (2) at (-4, 0.25) {};
		\node [style=none] (3) at (-4.75, 0.25) {};
		\node [style=none] (5) at (-3.875, -0.75) {};
		\node [style=X dot] (6) at (-3.125, -0.25) {};
		\node [style=none] (7) at (-2.375, -0.25) {};
		\node [style=none] (8) at (0.875, 0.5) {};
		\node [style=Z dot] (9) at (-0.375, 0.5) {};
		\node [style=none] (10) at (-0.875, 0.5) {};
		\node [style=none] (11) at (-0.875, -0.5) {};
		\node [style=X dot] (12) at (0.375, -0.5) {};
		\node [style=none] (13) at (0.875, -0.5) {};
		\node [style=none] (14) at (-1.625, 0) {$=$};
		\node [style=0 control] (15) at (3.875, 0.5) {$\Lambda$};
		\node [style=box] (16) at (3.875, -0.5) {$X$};
		\node [style=none] (17) at (4.875, 0.5) {};
		\node [style=none] (19) at (2.875, 0.5) {};
		\node [style=none] (20) at (2.875, -0.5) {};
		\node [style=none] (22) at (4.875, -0.5) {};
		\node [style=none] (23) at (1.875, 0) {$=$};
	\end{pgfonlayer}
	\begin{pgfonlayer}{edgelayer}
		\draw [style=simple] (0.center) to (2);
		\draw [style=simple] (3.center) to (2);
		\draw [style=simple] (7.center) to (6);
		\draw [style=simple] (6) to (5.center);
		\draw [style=simple] (2) to (6);
		\draw [style=simple] (8.center) to (9);
		\draw [style=simple] (10.center) to (9);
		\draw [style=simple] (13.center) to (12);
		\draw [style=simple] (12) to (11.center);
		\draw [style=simple] (9) to (12);
		\draw [style=simple] (15) to (16);
		\draw [style=simple] (19.center) to (15);
		\draw [style=simple] (17.center) to (15);
		\draw [style=simple] (16) to (22.center);
		\draw [style=simple] (16) to (20.center);
	\end{pgfonlayer}
\end{tikzpicture}
\end{equation}
\begin{equation}
    \begin{tikzpicture}
	\begin{pgfonlayer}{nodelayer}
		\node [style=none] (0) at (-3.875, 0.75) {};
		\node [style=Z dot] (2) at (-3.125, 0.25) {};
		\node [style=none] (3) at (-2.375, 0.25) {};
		\node [style=none] (5) at (-3.25, -0.75) {};
		\node [style=X dot] (6) at (-4, -0.25) {};
		\node [style=none] (7) at (-4.75, -0.25) {};
		\node [style=none] (8) at (-0.875, 0.5) {};
		\node [style=Z dot] (9) at (0.375, 0.5) {};
		\node [style=none] (10) at (0.875, 0.5) {};
		\node [style=none] (11) at (0.875, -0.5) {};
		\node [style=X dot] (12) at (-0.375, -0.5) {};
		\node [style=none] (13) at (-0.875, -0.5) {};
		\node [style=none] (14) at (-1.625, 0) {$=$};
		\node [style=0 control] (15) at (3.875, 0.5) {$\Lambda$};
		\node [style=box] (16) at (3.875, -0.5) {$X^{\dagger}$};
		\node [style=none] (17) at (4.875, 0.5) {};
		\node [style=none] (19) at (2.875, 0.5) {};
		\node [style=none] (20) at (2.875, -0.5) {};
		\node [style=none] (22) at (4.875, -0.5) {};
		\node [style=none] (23) at (1.875, 0) {$=$};
	\end{pgfonlayer}
	\begin{pgfonlayer}{edgelayer}
		\draw [style=simple] (0.center) to (2);
		\draw [style=simple] (3.center) to (2);
		\draw [style=simple] (7.center) to (6);
		\draw [style=simple] (6) to (5.center);
		\draw [style=simple] (2) to (6);
		\draw [style=simple] (8.center) to (9);
		\draw [style=simple] (10.center) to (9);
		\draw [style=simple] (13.center) to (12);
		\draw [style=simple] (12) to (11.center);
		\draw [style=simple] (9) to (12);
		\draw [style=simple] (15) to (16);
		\draw [style=simple] (19.center) to (15);
		\draw [style=simple] (17.center) to (15);
		\draw [style=simple] (16) to (22.center);
		\draw [style=simple] (16) to (20.center);
	\end{pgfonlayer}
\end{tikzpicture}
\end{equation}

Generalizing the two-qubit circuit for any controlled rotation, to (certain controlled rotations in) any qudit dimension:
\begin{equation}\label{eq:cphase}
    \begin{tikzpicture}
	\begin{pgfonlayer}{nodelayer}
		\node [style=none] (357) at (-2.775, 0.5) {};
		\node [style=none] (358) at (1.875, -0.5) {};
		\node [style=none] (359) at (1.875, 0.5) {};
		\node [style=none] (360) at (-2.775, -0.5) {};
		\node [style=Z dot] (363) at (-1.375, 0.5) {};
		\node [style=X dot] (364) at (-0.875, -0.5) {};
		\node [style=Z dot] (365) at (1.375, 0.5) {};
		\node [style=X dot] (366) at (0.875, -0.5) {};
		\node [style=Z phase dot] (367) at (-0.025, -0.5) {$\vec{\minu \alpha}$};
		\node [style=Z phase dot] (379) at (-2.175, -0.5) {$\vec{\alpha}$};
		\node [style=Z phase dot] (380) at (-2.175, 0.5) {$\vec{\alpha}$};
	\end{pgfonlayer}
	\begin{pgfonlayer}{edgelayer}
		\draw [style=simple] (364) to (363);
		\draw [style=simple] (359.center) to (357.center);
		\draw [style=simple] (366) to (365);
		\draw [style=simple] (360.center) to (358.center);
	\end{pgfonlayer}
\end{tikzpicture}
\end{equation}
It has a \emph{phase gadget}, a diagonal, symmetric two-qudit interaction:
\begin{equation}\label{eq:phasegadget}
    \begin{tikzpicture}
	\begin{pgfonlayer}{nodelayer}
		\node [style=none] (12) at (-6.875, -0.5) {};
		\node [style=X dot] (13) at (-5.875, -0.5) {};
		\node [style=none] (14) at (-3.375, -0.5) {};
		\node [style=Z dot] (15) at (-6.375, 0.5) {};
		\node [style=none] (16) at (-6.875, 0.5) {};
		\node [style=none] (17) at (-3.375, 0.5) {};
		\node [style=Z phase dot] (18) at (-5.125, -0.5) {$\vec{\minu \alpha}$};
		\node [style=X dot] (19) at (-4.375, -0.5) {};
		\node [style=Z dot] (20) at (-3.875, 0.5) {};
		\node [style=none] (21) at (-2.375, 0) {=};
		\node [style=none] (22) at (-1.375, -0.5) {};
		\node [style=X dot] (23) at (-0.125, -0.5) {};
		\node [style=none] (24) at (1.375, -0.5) {};
		\node [style=Z dot] (25) at (-0.875, 0.5) {};
		\node [style=none] (26) at (-1.375, 0.5) {};
		\node [style=none] (27) at (1.375, 0.5) {};
		\node [style=Z dot] (28) at (0.375, 0) {};
		\node [style=X dot] (29) at (0.875, -0.5) {};
		\node [style=Z phase dot] (30) at (1.375, 0) {$\vec{\minu \alpha}$};
		\node [style=none] (31) at (2.375, 0) {$=$};
		\node [style=Z dot] (41) at (-0.625, 0) {};
		\node [style=none] (55) at (3.375, -0.5) {};
		\node [style=none] (56) at (5.625, -0.5) {};
		\node [style=Z dot] (59) at (4.125, 0.5) {};
		\node [style=none] (60) at (3.375, 0.5) {};
		\node [style=none] (61) at (5.625, 0.5) {};
		\node [style=Z phase dot] (80) at (5.625, 0) {$\vec{\minu \alpha}$};
		\node [style=none] (82) at (5.625, -0.5) {};
		\node [style=none] (83) at (3.375, -0.5) {};
		\node [style=X dot] (84) at (4.75, 0.125) {};
		\node [style=Z dot] (85) at (4.5, -0.375) {};
	\end{pgfonlayer}
	\begin{pgfonlayer}{edgelayer}
		\draw (12.center) to (13);
		\draw (13) to (14.center);
		\draw (16.center) to (17.center);
		\draw (15) to (13);
		\draw (20) to (19);
		\draw (26.center) to (27.center);
		\draw (28) to (30);
		\draw (25) to (41);
		\draw (41) to (23);
		\draw (41) to (29);
		\draw (23) to (28);
		\draw (28) to (29);
		\draw (60.center) to (61.center);
		\draw [bend right=15, looseness=0.75] (82.center) to (85);
		\draw [bend left=15, looseness=0.50] (83.center) to (85);
		\draw (85) to (84);
		\draw [bend right=15] (84) to (80);
		\draw [bend left=15] (84) to (59);
		\draw (22.center) to (23);
		\draw (29) to (24.center);
	\end{pgfonlayer}
\end{tikzpicture}
\end{equation}
where above we applied the qudit bialgebra rule:
\begin{equation}\label{eq:bialgebra}
    \begin{tikzpicture}
	\begin{pgfonlayer}{nodelayer}
		\node [style=Z dot] (0) at (-2, 0.25) {};
		\node [style=Z dot] (1) at (-3, 0.25) {};
		\node [style=X dot] (2) at (-1.5, -0.25) {};
		\node [style=X dot] (3) at (-2.5, -0.25) {};
		\node [style=none] (4) at (-1, 0.5) {};
		\node [style=none] (5) at (-3.5, 0.5) {};
		\node [style=none] (6) at (-1, -0.5) {};
		\node [style=none] (7) at (-3.5, -0.5) {};
		\node [style=none] (8) at (0, 0) {$=$};
		\node [style=none] (9) at (2.625, 0.5) {};
		\node [style=none] (10) at (1, 0.5) {};
		\node [style=none] (11) at (2.625, -0.5) {};
		\node [style=none] (12) at (1, -0.5) {};
		\node [style=X dot] (13) at (2.125, 0.125) {};
		\node [style=Z dot] (14) at (1.625, -0.25) {};
	\end{pgfonlayer}
	\begin{pgfonlayer}{edgelayer}
		\draw (3) to (0);
		\draw (2) to (1);
		\draw [bend right=15, looseness=0.75] (1) to (3);
		\draw [bend right=15, looseness=0.75] (2) to (0);
		\draw (0) to (4.center);
		\draw (1) to (5.center);
		\draw [bend right=15] (11.center) to (14);
		\draw [bend left=15] (12.center) to (14);
		\draw (14) to (13);
		\draw [bend left=15] (13) to (10.center);
		\draw [bend right=15] (13) to (9.center);
		\draw (2) to (6.center);
		\draw (3) to (7.center);
	\end{pgfonlayer}
\end{tikzpicture}
\end{equation}
For qubits, any diagonal qubit unitary can be expressed as a product of phase gadgets by writing the unitary as a phase polynomial~\cite{AmyM2019verification,vandeGriendAM2020architecturephasepoly}, as it serves as a good basis for optimising quantum circuits~\cite{CowtanA2020phasegadget,CowtanA2020uccansatz,deBeaudrapN2020reducepifourphase,deBeaudrapN2020treducspidernest,vandeWeteringJ2021globalgates,Backens2020extraction}. It may thus inform how to optimize qudit computation, building upon the work for qutrits in Ref.~\cite{TeagueAT2022qutritsimp}.

There are several different ways to define a two-qudit phase gadget.
One way is to consider it as the diagonal gate $\ket{x,y}\mapsto \omega^{\alpha (x\oplus y)} \ket{x,y}$. This applies a phase of $\omega^{\alpha (x\oplus y)}$ when $x\oplus y = 1$. Here $\oplus$ is the XOR operation, which is the addition on $\mathbb{Z}_d$.

Inputting to Equation~\eqref{eq:phasegadget} Z-basis states $\ket{x,y}$ (X spiders with labels $\vec{x} = [(d\minu 1)x,...,x]$ and $\vec{y} = [(d\minu 1)y,...,y]$) evaluates to:
\begin{equation}\label{eq:phase-gadgets-qudit-calc}
    \begin{tikzpicture}
	\begin{pgfonlayer}{nodelayer}
		\node [style=none] (44) at (-2.5, 0) {$=$};
		\node [style=X phase dot] (45) at (-6.25, -0.75) {\tiny$\vec{y}$};
		\node [style=none] (46) at (-4, -0.75) {};
		\node [style=Z dot] (47) at (-5.25, 0.75) {};
		\node [style=X phase dot] (48) at (-6.25, 0.75) {\tiny$\vec{x}$};
		\node [style=none] (49) at (-4, 0.75) {};
		\node [style=Z phase dot] (50) at (-3.65, 0) {$\vec{\minu \alpha}$};
		\node [style=X dot] (51) at (-4.5, 0) {};
		\node [style=Z dot] (52) at (-5.25, -0.75) {};
		\node [style=X phase dot] (54) at (0, -0.75) {\tiny$\vec{y}$};
		\node [style=none] (55) at (1, -0.75) {};
		\node [style=X phase dot] (57) at (0, 0.75) {\tiny$\vec{x}$};
		\node [style=none] (58) at (1, 0.75) {};
		\node [style=Z phase dot] (59) at (1.35, 0) {$\vec{\minu \alpha}$};
		\node [style=X dot] (60) at (0, 0) {};
		\node [style=X phase dot] (63) at (-1.15, 0.5) {\tiny$\vec{x}$};
		\node [style=X phase dot] (64) at (-1.15, -0.5) {\tiny$\vec{y}$};
		\node [style=none] (76) at (2.5, 0) {=};
		\node [style=X phase dot] (77) at (5.575, -0.75) {\tiny$\vec{y}$};
		\node [style=none] (78) at (6.575, -0.75) {};
		\node [style=X phase dot] (79) at (5.575, 0.75) {\tiny$\vec{x}$};
		\node [style=none] (80) at (6.575, 0.75) {};
		\node [style=Z phase dot] (81) at (4.925, 0) {$\vec{\minu \alpha}$};
		\node [style=X phase dot] (84) at (3.675, 0) {\tiny$\vec{x+y}$};
	\end{pgfonlayer}
	\begin{pgfonlayer}{edgelayer}
		\draw (48) to (49.center);
		\draw (51) to (47);
		\draw (51) to (50);
		\draw (52) to (45);
		\draw (52) to (46.center);
		\draw (52) to (51);
		\draw (57) to (58.center);
		\draw [in=180, out=0] (60) to (59);
		\draw (54) to (55.center);
		\draw (63) to (60);
		\draw (64) to (60);
		\draw (79) to (80.center);
		\draw (77) to (78.center);
		\draw (84) to (81);
	\end{pgfonlayer}
\end{tikzpicture}
\end{equation}

Next, we interpret each spider in the rightmost diagram as its corresponding linear map: the identity on Z-basis states accompanied by a `floating scalar' expression that evaluates to $\omega^{\minu \alpha (x+y \text{ mod }d)}$.  Hence this phase gadget indeed implements the operation we want.  Therefore, we see that these three ways to define a qubit phase gadget---via the action, via the circuit, or via the diagrammatic representation---are equal in all finite qudit dimensions.

For our first controlled $H$ gate construction, we seek to implement a controlled phase gate such that the phase gate $\vec{\tau}$ applies to the target when $x = 1$, and identity applies to the target when $x = 0$.  We can thus solve the below system of equations for $\vec{\alpha}$ to find when the circuit of Equation~\eqref{eq:cphase} implements $\sqrt[\left(2^{\minu 1}\,\,\text{mod}\,\,d\right)]{Z(\vec{\tau})}$:
\begin{equation}
    \begin{cases}
        \alpha(x) + \alpha(y) - \alpha(x+y\text{ mod }d) = 0 \ &\text{if}\ x = 0 \\
        \alpha(x) + \alpha(y) - \alpha(x+y\text{ mod }d) = y^2\text{ mod }d \ &\text{if}\ x = 1
    \end{cases}
\end{equation}
Setting $\alpha(0) = 0$ and solving for the remaining elements of $\vec{\alpha}$:
\begin{equation}
    \alpha(k) = \minu \sum_{j=0}^{k\minu 1} j^2 \text{ mod }d
\end{equation}
To implement $Z(\vec{\alpha})$, as a diagonal gate comprised solely of integer powers of $\omega$, we can always construct it from $P_1(k)$ gates decomposed in Equation~\eqref{eq:Pk}.  Then we can build controlled $H$, as applying the $Z(\vec{\alpha})$ gate $2^{\minu 1}\,\,\text{mod}\,\,d$ times gives $Z(\vec{\tau})$.  Since $\alpha(1) = 0$, we can omit the $Z(\vec{\tau})$ gate on the control of the phase gadget:
\begin{lemma}\label{lem:och}
    The following circuit performs identity on the target when the control is $\ket{0}$, and the $H$ gate up to a controlled global phase of $i^{(d\minu 1)/2}$ when the control is $\ket{1}$:
\begin{equation}\label{eq:och}
    \begin{tikzpicture}
	\begin{pgfonlayer}{nodelayer}
		\node [style=none] (357) at (-6, 0.75) {};
		\node [style=none] (358) at (5, -0.75) {};
		\node [style=none] (359) at (5, 0.75) {};
		\node [style=none] (360) at (-6, -0.75) {};
		\node [style=box] (386) at (-2.625, -0.75) {$H$};
		\node [style=box] (387) at (1.25, -0.75) {$H^{\dagger}$};
		\node [style=Z dot] (400) at (-5.125, 0.75) {};
		\node [style=Z phase dot] (401) at (-5.125, -0.75) {$\vec{\tau}$};
		\node [style=X dot] (402) at (-4.625, 0) {};
		\node [style=Z phase dot] (403) at (-3.875, 0) {$\vec{\minu \tau}$};
		\node [style=Z dot] (404) at (-1.375, 0.75) {};
		\node [style=Z phase dot] (405) at (-1.375, -0.75) {$\vec{\tau}$};
		\node [style=X dot] (406) at (-0.875, 0) {};
		\node [style=Z phase dot] (407) at (-0.125, 0) {$\vec{\minu \tau}$};
		\node [style=Z dot] (408) at (2.875, 0.75) {};
		\node [style=Z phase dot] (409) at (2.875, -0.75) {$\vec{\tau}$};
		\node [style=X dot] (410) at (3.375, 0) {};
		\node [style=Z phase dot] (411) at (4.125, 0) {$\vec{\minu \tau}$};
	\end{pgfonlayer}
	\begin{pgfonlayer}{edgelayer}
		\draw [style=simple] (359.center) to (357.center);
		\draw [style=simple, in=180, out=0] (360.center) to (358.center);
		\draw [style=simple] (400) to (402);
		\draw [style=simple] (403) to (402);
		\draw [style=simple] (402) to (401);
		\draw [style=simple] (404) to (406);
		\draw [style=simple] (407) to (406);
		\draw [style=simple] (406) to (405);
		\draw [style=simple] (408) to (410);
		\draw [style=simple] (411) to (410);
		\draw [style=simple] (410) to (409);
	\end{pgfonlayer}
\end{tikzpicture}
\end{equation}
\end{lemma}

\begin{proposition}
    With a clean and uncomputed ancilla, we can adapt the construction of Lemma~\ref{lem:och} to implement the $\ket{1}$-controlled $H$ gate up to a controlled global phase:
\begin{equation}
    \begin{tikzpicture}
	\begin{pgfonlayer}{nodelayer}
		\node [style=none] (357) at (-7, 0) {};
		\node [style=none] (358) at (8.25, -1.5) {};
		\node [style=none] (359) at (7, 0) {};
		\node [style=none] (360) at (-8.25, -1.5) {};
		\node [style=none] (394) at (-8.25, 1) {};
		\node [style=none] (395) at (8.25, 1) {};
		\node [style=box] (402) at (-6.125, 0) {$X$};
		\node [style=box] (404) at (5.875, 0) {$X^{\dagger}$};
		\node [style=0 control] (405) at (-6.125, 1) {\textsf{1}};
		\node [style=none] (406) at (-7.5, 0) {$|0\rangle$};
		\node [style=none] (408) at (7.5, 0) {$\langle 0|$};
		\node [style=0 control] (409) at (5.875, 1) {\textsf{1}};
		\node [style=box] (421) at (-2.375, -1.5) {$H$};
		\node [style=box] (422) at (1.5, -1.5) {$H^{\dagger}$};
		\node [style=Z dot] (423) at (-4.875, 0) {};
		\node [style=Z phase dot] (424) at (-4.875, -1.5) {$\vec{\tau}$};
		\node [style=X dot] (425) at (-4.375, -0.75) {};
		\node [style=Z phase dot] (426) at (-3.625, -0.75) {$\vec{\minu \tau}$};
		\node [style=Z dot] (427) at (-1.125, 0) {};
		\node [style=Z phase dot] (428) at (-1.125, -1.5) {$\vec{\tau}$};
		\node [style=X dot] (429) at (-0.625, -0.75) {};
		\node [style=Z phase dot] (430) at (0.125, -0.75) {$\vec{\minu \tau}$};
		\node [style=Z dot] (431) at (3.125, 0) {};
		\node [style=Z phase dot] (432) at (3.125, -1.5) {$\vec{\tau}$};
		\node [style=X dot] (433) at (3.625, -0.75) {};
		\node [style=Z phase dot] (434) at (4.375, -0.75) {$\vec{\minu \tau}$};
	\end{pgfonlayer}
	\begin{pgfonlayer}{edgelayer}
		\draw [style=simple] (359.center) to (357.center);
		\draw [style=simple, in=180, out=0] (360.center) to (358.center);
		\draw [style=simple] (395.center) to (394.center);
		\draw [style=simple] (405) to (402);
		\draw [style=simple] (409) to (404);
		\draw [style=simple] (423) to (425);
		\draw [style=simple] (426) to (425);
		\draw [style=simple] (425) to (424);
		\draw [style=simple] (427) to (429);
		\draw [style=simple] (430) to (429);
		\draw [style=simple] (429) to (428);
		\draw [style=simple] (431) to (433);
		\draw [style=simple] (434) to (433);
		\draw [style=simple] (433) to (432);
	\end{pgfonlayer}
\end{tikzpicture}
\end{equation}
\end{proposition}
\begin{remark}
    For all $W$ states in this text, we conserve non-Clifford gate count by not correcting this controlled global phase; these controlled global phases collectively incur only a global phase for the final $W$ state. Moreover, whenever the control qubit is known to be in the $\{\ket{0}, \ket{1}\}$ subspace, for instance in Equation~\eqref{eq:Wcircspread}, the ancilla-free implementation in Equation~\eqref{eq:och} suffices.
\end{remark}

\subsection{Prime power size \texorpdfstring{$W$}{W} states}\label{sec:primepowerW}
The goal is to now scale the size of the generalized qubit $W$ state $W_{d^n}$ synthesized using qudit gates, to any prime power size.
This is accomplished by implementing a gate which has the action
\begin{equation}\label{eq:spreadcrit}
    \ket{0...0} \mapsto e^{i \theta} \ket{0...0} \quad \text{ and } \quad \ket{10...0} \mapsto e^{i \phi} W_{d^n}
\end{equation}
where $e^{i \theta}$ and $e^{i \phi}$ are any global phases.

\begin{lemma}
    The $\textsf{spread}$ gate, a $d$-qudit gate which fulfills the criteria of Equation~\eqref{eq:spreadcrit}, is implementable by the circuit
    \begin{equation}\label{eq:Wcircspread}
        \begin{tikzpicture}
	\begin{pgfonlayer}{nodelayer}
		\node [style=none] (240) at (-5.75, 1) {};
		\node [style=none] (241) at (-5.75, 0) {};
		\node [style=none] (246) at (-5.75, 2) {};
		\node [style=none] (295) at (8.125, 0) {};
		\node [style=none] (296) at (8.125, -1) {};
		\node [style=none] (297) at (8.125, -2) {};
		\node [style=none] (299) at (8.125, 2) {};
		\node [style=none] (300) at (8.125, 1) {};
		\node [style=none] (305) at (-5.75, -1) {};
		\node [style=none] (306) at (-5.75, -2) {};
		\node [style=0 control] (310) at (-4.875, 2) {\textsf{1}};
		\node [style=box] (311) at (-4.875, 1) {$H$};
		\node [style=0 control] (314) at (-2, 1) {\textsf{0}};
		\node [style=box] (315) at (-2, 2) {$+1$};
		\node [style=box] (316) at (-3.5, 2) {$\minu 1$};
		\node [style=0 control] (317) at (-1, 1) {\textsf{2}};
		\node [style=box] (318) at (-1, 0) {$+1$};
		\node [style=0 control] (319) at (0.25, 0) {$\Lambda$};
		\node [style=box] (320) at (0.25, 1) {$\minu 2$};
		\node [style=0 control] (321) at (1.375, 1) {\textsf{3}};
		\node [style=box] (322) at (1.375, -1) {$+1$};
		\node [style=0 control] (323) at (2.625, -1) {$\Lambda$};
		\node [style=box] (324) at (2.625, 1) {$\minu 3$};
		\node [style=0 control] (325) at (3.75, 1) {\textsf{4}};
		\node [style=box] (326) at (3.75, -2) {$+1$};
		\node [style=0 control] (327) at (5, -2) {$\Lambda$};
		\node [style=box] (328) at (5, 1) {$\minu 4$};
		\node [style=none] (329) at (5.5, -2.75) {$\rotatebox{135}{...}$};
		\node [style=none] (335) at (-2.75, 2.75) {};
		\node [style=none] (336) at (-2.75, -3.25) {};
		\node [style=none] (337) at (7.875, 2.75) {};
		\node [style=none] (338) at (7.875, -3.25) {};
		\node [style=none] (339) at (-6.375, 0) {$\coloneqq$};
		\node [style=none] (341) at (-8.5, 2.25) {};
		\node [style=none] (342) at (-7.5, 2.25) {};
		\node [style=none] (343) at (-8.5, -3) {};
		\node [style=none] (344) at (-7.5, -3) {};
		\node [style=none] (345) at (-8.875, 1) {};
		\node [style=none] (346) at (-8.875, 0) {};
		\node [style=none] (347) at (-8.875, 2) {};
		\node [style=none] (348) at (-8.875, -1) {};
		\node [style=none] (349) at (-8.875, -2) {};
		\node [style=none] (350) at (-8.5, 1) {};
		\node [style=none] (351) at (-8.5, 0) {};
		\node [style=none] (352) at (-8.5, 2) {};
		\node [style=none] (353) at (-8.5, -1) {};
		\node [style=none] (354) at (-8.5, -2) {};
		\node [style=none] (355) at (-7.5, 1) {};
		\node [style=none] (356) at (-7.5, 0) {};
		\node [style=none] (357) at (-7.5, 2) {};
		\node [style=none] (358) at (-7.5, -1) {};
		\node [style=none] (359) at (-7.5, -2) {};
		\node [style=none] (360) at (-7.125, 1) {};
		\node [style=none] (361) at (-7.125, 0) {};
		\node [style=none] (362) at (-7.125, 2) {};
		\node [style=none] (363) at (-7.125, -1) {};
		\node [style=none] (364) at (-7.125, -2) {};
		\node [style=none] (365) at (-8, -0.25) {$\rotatebox{90}{\tiny\textsf{spread}}$};
		\node [style=none] (366) at (-7.25, -2.5) {$\rotatebox{90}{...}$};
		\node [style=none] (367) at (-8.75, -2.5) {$\rotatebox{90}{...}$};
		\node [style=box] (368) at (6.75, 1) {\large$\,\,\,\,\,\,\,\,\,\,\,\,$};
		\node [style=none] (369) at (6.75, 1) {${Z^{\dagger}}^{\frac{d\minu 1}{2}}$};
	\end{pgfonlayer}
	\begin{pgfonlayer}{edgelayer}
		\draw [style=simple, in=270, out=90] (311) to (310);
		\draw [style=simple] (315) to (314);
		\draw [style=simple] (318) to (317);
		\draw [style=simple] (320) to (319);
		\draw [style=simple] (322) to (321);
		\draw [style=simple] (324) to (323);
		\draw [style=simple] (326) to (325);
		\draw [style=simple] (328) to (327);
		\draw [style=simple] (246.center) to (299.center);
		\draw [style=simple] (300.center) to (240.center);
		\draw [style=simple] (241.center) to (295.center);
		\draw [style=simple] (296.center) to (305.center);
		\draw [style=simple] (306.center) to (297.center);
		\draw [style=light blue line] (336.center) to (338.center);
		\draw [style=light blue line] (338.center) to (337.center);
		\draw [style=light blue line] (337.center) to (335.center);
		\draw [style=light blue line] (335.center) to (336.center);
		\draw [style=simple] (341.center) to (342.center);
		\draw [style=simple] (342.center) to (344.center);
		\draw [style=simple] (344.center) to (343.center);
		\draw [style=simple] (341.center) to (343.center);
		\draw [style=simple] (352.center) to (347.center);
		\draw [style=simple] (350.center) to (345.center);
		\draw [style=simple] (351.center) to (346.center);
		\draw [style=simple] (353.center) to (348.center);
		\draw [style=simple] (354.center) to (349.center);
		\draw [style=simple] (364.center) to (359.center);
		\draw [style=simple] (363.center) to (358.center);
		\draw [style=simple] (361.center) to (356.center);
		\draw [style=simple] (360.center) to (355.center);
		\draw [style=simple] (362.center) to (357.center);
	\end{pgfonlayer}
\end{tikzpicture}
    \end{equation}
\end{lemma}
Note that the blue box contains the circuit in Equation~\eqref{eq:Wcirc}.

\begin{proposition}\label{prop:spreadout}
    By recursively applying the $\textsf{spread}$ gate to each output qupit tensored with $d\minu 1$ $\ket{0}$ qupits, for a total of $n$ layers, the entanglement can be uniformly spread to build the $N=d^n$ size $W$ state up to a global phase with $O(N)$ $\sqrt[d]{Z}$ count.
    \begin{equation}\label{eq:spreadout}
        \input{spreadout.tikz}
    \end{equation}
    Substitute the first $\textsf{spread}$ gate by Equation~\eqref{eq:Wcirc} to reduce gate count.
\end{proposition}
\begin{proof}
    The $i^{\text{th}}$ layer of $\textsf{spread}$ gates consists of $d^{i\minu 1}$ $\textsf{spread}$ gates.  Therefore, the total number of $\textsf{spread}$ gates is
    \begin{equation}
        \sum_{j=0}^{n\minu 1} d^{j} = \frac{d^n \minu 1}{d \minu 1} = O(N)
    \end{equation}
    As each $\textsf{spread}$ gate has gate count independent of $N$, the total $\sqrt[d]{Z}$ count is $O(N)$.
\end{proof}

\begin{remark}
    The qubit version of the above method was discovered, albeit not well-known, in 2018 in a StackExchange post by Gidney~\cite{GidneyC2018Wstatemagic}.  It achieves a $T$-count of $2N\minu 4$, establishing an upper bound for the minimum number of $T$ gates needed to implement the $W$ state.  Since then, the asymptotic lower bound has been found to match the upper bound: based on a reasonable complexity-theoretic conjecture, Arunachalam, Bravyi, Nirkhe, and O'Gorman found that any qubit Clifford+$T$ circuit building an $N$-qubit $W$-state must contain at least a number of $T$ gates linear in $N$~\cite{ArunachalamS2022WstateTcount}.  It would be interesting to investigate the non-Clifford resource requirements of the generalized qubit and qudit $W$ states in the qupit setting.
\end{remark}

\subsubsection{Prime power size qupit \texorpdfstring{$W$}{W} states}
Furthermore, we generalize the constructions thus far from \emph{qubit} $W$ states to \emph{qudit} $W$ states.
Consider this modification of the $\textsf{spread}$ circuit in Equation~\eqref{eq:Wcircspread}:
\begin{equation}\label{eq:Wcircspreadqudit}
    \begin{tikzpicture}
	\begin{pgfonlayer}{nodelayer}
		\node [style=none] (240) at (-7.75, 0.5) {};
		\node [style=none] (241) at (-7.75, -0.5) {};
		\node [style=none] (246) at (-7.75, 1.5) {};
		\node [style=none] (295) at (6.75, -0.5) {};
		\node [style=none] (296) at (6.75, -1.5) {};
		\node [style=none] (299) at (6.75, 1.5) {};
		\node [style=none] (300) at (6.75, 0.5) {};
		\node [style=none] (305) at (-7.75, -1.5) {};
		\node [style=0 control] (310) at (-6.75, 1.5) {\textsf{1}};
		\node [style=box] (311) at (-6.75, 0.5) {$H$};
		\node [style=0 control] (314) at (-5.25, 0.5) {\textsf{0}};
		\node [style=box] (315) at (-5.25, 1.5) {$+1$};
		\node [style=box] (316) at (-3.75, 1.5) {$-1$};
		\node [style=none] (329) at (6.25, -2.25) {$\rotatebox{135}{...}$};
		\node [style=0 control] (340) at (-2.25, 1.5) {\textsf{2}};
		\node [style=box] (341) at (-2.25, -0.5) {$H$};
		\node [style=0 control] (342) at (-0.75, -0.5) {\textsf{0}};
		\node [style=box] (343) at (-0.75, 1.5) {$+2$};
		\node [style=box] (344) at (0.75, 1.5) {$-2$};
		\node [style=0 control] (345) at (2.25, 1.5) {\textsf{3}};
		\node [style=box] (346) at (2.25, -1.5) {$H$};
		\node [style=0 control] (347) at (3.75, -1.5) {\textsf{0}};
		\node [style=box] (348) at (3.75, 1.5) {$+3$};
		\node [style=box] (349) at (5.25, 1.5) {$-3$};
	\end{pgfonlayer}
	\begin{pgfonlayer}{edgelayer}
		\draw [style=simple, in=270, out=90] (311) to (310);
		\draw [style=simple] (315) to (314);
		\draw [style=simple] (246.center) to (299.center);
		\draw [style=simple] (300.center) to (240.center);
		\draw [style=simple] (241.center) to (295.center);
		\draw [style=simple] (296.center) to (305.center);
		\draw [style=simple, in=270, out=90] (341) to (340);
		\draw [style=simple] (343) to (342);
		\draw [style=simple, in=270, out=90] (346) to (345);
		\draw [style=simple] (348) to (347);
	\end{pgfonlayer}
\end{tikzpicture}
\end{equation}
For inputs of the form $\ket{k0...0}$ for $k \in \mathbb{Z}_d$, up to phases we can safely ignore, it has action
\begin{equation}
    \begin{cases}
        \ket{k0...0} \mapsto \ket{k0...0} \ &\text{if }\ k = 0 \\
        \ket{k0...0} \mapsto \frac{1}{\sqrt{d}} \left( \sum_{j=1}^{d\minu 1} \ket{0}^{\otimes k} \ket{j} \ket{0}^{\otimes (d\minu k\minu 1)} + \ket{k0...0} \right) \ &\text{if }\ k \neq 0
    \end{cases}
\end{equation}

The $d$-qudit $W$ state is achieved by inputting to the top qudit of the above circuit the resource state
\begin{equation}
    \begin{tikzpicture}
	\begin{pgfonlayer}{nodelayer}
		\node [style=none] (0) at (2, -1) {};
		\node [style=none] (2) at (4, -1) {};
		\node [style=none] (24) at (2, 0) {};
		\node [style=none] (25) at (5.5, 0) {};
		\node [style=box] (26) at (3, -1) {$X$};
		\node [style=0 control] (28) at (3, 0) {\textsf{0}};
		\node [style=none] (29) at (1.5, -1) {$|0\rangle$};
		\node [style=none] (30) at (4.5, -1) {$\langle 0|$};
		\node [style=none] (38) at (0, 0) {$=$};
		\node [style=none] (39) at (-3, 0) {$\frac{1}{\sqrt{d-1}} \sum_{j=1}^{d-1} |j\rangle $};
		\node [style=none] (41) at (1.4, 0) {$|+\rangle$};
	\end{pgfonlayer}
	\begin{pgfonlayer}{edgelayer}
		\draw [style=simple] (2.center) to (0.center);
		\draw [style=simple] (25.center) to (24.center);
		\draw [style=simple] (28) to (26);
	\end{pgfonlayer}
\end{tikzpicture}
\end{equation}
and inputting $\ket{0}$ for all other qudits.

Moreover, as we did for qubit $W$ states, we can likewise exponentially magnify the size of this qudit $W$ state.  This is because the above circuit maps, for some irrelevant global phases $e^{i \theta}$ and $e^{i \phi}$,
\begin{equation}\label{eq:spreadcrit2}
    \ket{0...0} \mapsto e^{i \theta} \ket{0...0} \quad \text{ and } \quad \left(\frac{1}{\sqrt{d\minu 1}} \sum_{j=1}^{d\minu 1} \ket{j}\right)\ket{0...0} \mapsto e^{i \phi} W_{d^n}
\end{equation}
Therefore,
\begin{proposition}\label{prop:quditspreadout}
We can build the size $N = d^n$ qudit (of prime dimension $d$) $W$ state by inputting $\left(\frac{1}{\sqrt{d\minu 1}} \sum_{j=1}^{d\minu 1} \ket{j}\right)\ket{0...0}$ into $n$ layers of the circuit in Proposition~\ref{prop:spreadout}, where the qudit $W$ state $\textsf{spread}$ gate gate is given in Equation~\eqref{eq:Wcircspreadqudit}.
This has $O(N)$ $\sqrt[d]{Z}$ count because the $\sqrt[d]{Z}$ count for each of the $O(N)$ $\textsf{spread}$ gates is independent of $N$ for fixed $d$.
This requires no further post-selection other than the initial copy of the above resource state.
\end{proposition}

\subsection{Building any size \texorpdfstring{$W$}{W} states}\label{sec:othersizeW}
In the fault-tolerant setting, protocols to generate large $W$ states should be both scalable, robust against errors, and deterministic.  The above protocols as-is generate only prime power size $W$ states.
We now present two ways to construct any size qubit $W$ state.

\subsubsection{By switching dimensions mid-computation}
Here, we modify earlier constructions to deterministically build any composite (i.e. non-prime) size $N$-qubit $W$ states.
This is by generalizing Proposition~\ref{prop:spreadout} to allow $\textsf{spread}$ gates of differing qudit dimensions.

\begin{proposition}\label{prop:nonprimeW}
    For any composite number $N$ with prime factorization $N = N_1 \times N_2 \times ... \times N_f$ (which is unique up to permutation), a deterministic circuit to build the $N$-qubit $W$ state is the circuit of Proposition~\ref{prop:spreadout}, where the $i^{\text{th}}$ layer of $\textsf{spread}$ gates is of qudit dimension $N_i$.
    This has $O(N)$ non-Clifford gate count because the $\sqrt[\text{\tiny $\raisebox{1pt}{$N_i$}$}]{Z}$ count for each of the $O(N)$ $\textsf{spread}(N_i)$ gates is independent of $N$.
\end{proposition}

\begin{remark}
    Even if for any prime qudit dimension $d$ the\linebreak Clifford+$\sqrt[d]{Z}$ gate set were efficiently implementable fault-tolerantly, for the circuit in Proposition~\ref{prop:nonprimeW} to be efficiently implementable fault-tolerantly, some mechanism to compose different qudit dimension states or gates is necessary.  One solution would be to find a \emph{code switching} protocol~\cite{Kubica2018TheAO} which can switch mid-computation from the quantum error correcting code in use to another of a different qudit dimension.  A framework to consider for computations involving differing qudit dimensions is the \emph{qufinite ZX-calculus}~\cite{WangQ2021qufinite}.
\end{remark}

\begin{example}
    Three deterministic 6-qubit $W$ state circuits:
    \begin{equation}
        \input{W6.tikz}
    \end{equation}
    The rightmost circuit above is moreover fault-tolerant: distill the qutrit magic state $\ket{+_2} = \frac{1}{\sqrt{2}} \left(\ket{0} + \ket{1}\right)$ (a qubit X-basis state) in the 5-qutrit stabilizer code~\cite{AnwarH2012qutritmsd,DawkinsH2015qutritmsdtight} or the 11-qutrit Golay code~\cite{PrakashS2020qutritgolay}, then do in a qutrit stabilizer code CX$^{\dagger}$ and two \textsf{spread(3)} gates.
\end{example}

\subsubsection{By post-selection of the $W$ state}
To attain other size $N$-qubit W states via post-selection, necessitates $d^{\ceil{\text{log}_d{N}}} \minu N$ qubits to be successfully post-selected on the $\ket{0}$ state.  If this fails, this protocol must be repeated until success; the expected number of attempts is
\begin{equation}
    \left( \prod_{j=N+1}^{d^{\ceil{\text{log}_d{N}}}}{\left(1 \minu \frac{1}{j}\right)} \right) \sum_{h=0}^{\infty}{(h+1) \left(1 \minu \prod_{k=N+1}^{d^{\ceil{\text{log}_d{N}}}}{\left(1 \minu \frac{1}{k}\right)}\right)^h} = \,\, \frac{d^{\ceil{\text{log}_d{N}}}}{N}.
\end{equation}
This is a geometric series of the form $(1\minu r) \sum_{h=0}^{\infty} (h+1) r^h = \frac{1}{1\minu r}$, where $1\minu r$ is a recurrence relation that solves to $\frac{N}{d^{\ceil{\text{log}_d{N}}}}$.  Hence for large $N$, the expected number of attempts is at most $d$.  In other words, the expected gate count to create an $N$-qubit $W$ state is at most a factor of $d$ greater than the optimal scenario where the post-selection succeeds on the first try.  Thus, the expected gate count remains linear in $N$.

\subsection{Considerations for network topologies}\label{sec:network}
The protocols in this work only use one- and two-qudit gates. It is easy to see that, in order to perform a two-qudit gate on two qudits which have a path length of $b$ nodes apart in the connectivity graph, an overhead of $O(b)$ swaps per two-qudit gate suffices to perform a two-qudit gate between them.  Therefore, the gate count of the protocol adapted to any quantum network topology is $O(bN)$.

Moreover, within each \textsf{spread} gate, the two-qudit gates you do are always less than $d$ distance apart.  Therefore, assuming nearest-neighbor connectivity in the network of each node (i.e. each logical qudit) participating in the protocol, the overhead due to routing can be bounded by a factor of $O(d)$.  Thus, the best possible performance of these protocols is when the network topology graph contains as a subgraph, the tree whose shape is that of the circuit in Equation~\eqref{eq:spreadout}.

\section{Conclusion}\label{sec:conclusion}
In this work, we first introduced the Clifford+$\sqrt[d]{Z}$ gate set and discussed its prospects for fault-tolerant quantum computation for arbitrary prime qudit dimension.  We then constructed the $\ket{0}$-controlled $X$ gate, which in forthcoming work we will show enables constructing any $d$-ary reversible classical gate.  We followed that with a construction of the controlled $H$ gate in any odd qudit dimension, which we used to construct boundlessly large qubit and qudit $W$ states of any prime power size for any prime qudit dimension.

The biggest open problem in relation to this work is the conjecture by deSilva that single-qudit gates in any level of the Clifford hierarchy can be implemented fault-tolerantly~\cite{deSilvaN2021qupitgateteleportation}.  If true, then the quantum operations available to us in the fault-tolerant regime is richer than we know it to be today.  For one, the fault-tolerant viability of emulation of binary operations on prime dimension qudits hinges on this question, as the cost of doing so increases with the qudit dimension $d$.  If the conjecture is false, then the interesting question arises of what operations can be done efficiently fault-tolerantly, more specifically characterizing this fragment of the logically useful constructions such as classical reversible gates and multipartite entangled resource states.

We are curious to see the problem of qudit $W$ state synthesis approached further by utilizing techniques from the study of resource theory of non-stabilizer states.  While constructions upper bound the minimum cost of synthesizing a resource state, computing the minimum `magic' (a measure of non-stabilizerness) inherent to a resource state can be effective at establishing a lower bound.  This may provide insight into comparison of deterministic (EQP) versus probabilistic (BQP) resource state synthesis.

Building upon these constructions, we aim to convert between photonic quantum computing and the circuit model.  The W algebra and its interaction with the Z algebra underpin the ZW-calculus which has applications especially in photonic quantum computing.  The translation of W state into ZX is a necessary prerequisite for proving completeness (that the equational theory supported by the graphical calculus is rich enough to derive all equalities of the underlying linear maps) of the qudit ZW-calculus~\cite{CoeckeB2011ThreeQubitEntanglement,hadzihasanovicDiagrammaticAxiomatisationQubit2015} from the very recent proof of qudit ZXW -calculus completeness~\cite{PoorB2023quditZXWcompleteness}.  This may facilitate investigations into reasoning involving controlled unitaries and Hamiltonian exponentiation~\cite{ShaikhRA2022zxwham}, multiplexors~\cite{HerrmannThesis}, and boson-fermion interaction ~\cite{HadzihasanovicA2018zxzwaxioms,deFeliceG2019fermioniccalc}.

\begin{acks}
    We would like to thank John van de Wetering for discussions on constructing qudit logical gates and feedback on the paper structure, Craig Gidney for insight into different ways of constructing $W$ states and qudit logical gates, Patrick Roy for discussions on qupit graphical calculi, Alexander Cowtan for reviewing the academic writing, Razin Shaikh for suggestions on the graphical presentation, and Tam\'as V\'ertesi for perspective on applications of multipartite entanglement.  LY is supported by an Oxford - Basil Reeve Graduate Scholarship at Oriel College with the Clarendon Fund.
\end{acks}

\bibliographystyle{ACM-Reference-Format}
\bibliography{main.bib}

\end{document}